\documentclass{article}

\usepackage{microtype}
\usepackage{graphicx}
\usepackage{subcaption}
\usepackage{booktabs} 
\usepackage{booktabs}
\usepackage{graphicx}
\usepackage[table]{xcolor} 
\usepackage{booktabs}
\usepackage{multirow}
\usepackage{hyperref}

\usepackage[preprint]{icml2026}
\usepackage{amsmath}
\usepackage{amssymb}
\usepackage{mathtools}
\usepackage{amsthm}
\usepackage[capitalize,noabbrev]{cleveref}
\theoremstyle{plain}
\newtheorem{theorem}{Theorem}[section]

\newtheorem{lemma}[theorem]{Lemma}
\newtheorem{corollary}[theorem]{Corollary}
\theoremstyle{definition}
\newtheorem{definition}[theorem]{Definition}
\newtheorem{assumption}[theorem]{Assumption}
\theoremstyle{remark}

\usepackage[textsize=tiny]{todonotes}
\icmltitlerunning{}
\usepackage{amsmath,amssymb,amsthm}
\usepackage{textcomp}
\makeatletter
\usepackage{wasysym} 
\DeclareUnicodeCharacter{266B}{\quarternote}
\@ifundefined{theorem}{\newtheorem{theorem}{Theorem}}{}
\@ifundefined{lemma}{\newtheorem{lemma}{Lemma}}{}
\@ifundefined{corollary}{\newtheorem{corollary}{Corollary}}{}
\@ifundefined{definition}{\newtheorem{definition}{Definition}}{}
\@ifundefined{assumption}{\newtheorem{assumption}{Assumption}}{}
\makeatother

\begin{document}
\icmltitlerunning{AtomicRAG: Atom--Entity Graphs for Retrieval-Augmented Generation}
\twocolumn[
  \icmltitle{AtomicRAG: Atom--Entity Graphs for Retrieval-Augmented Generation}
  \icmlsetsymbol{equal}{*}
    \icmlsetsymbol{corresponding}{\texorpdfstring{\textdagger}{dagger}}
  \begin{icmlauthorlist}
    \icmlauthor{Yanning Hou}{nudt,equal}
    \icmlauthor{Duanyang Yuan}{nudt,equal}
    \icmlauthor{Sihang Zhou}{nudt,corresponding}
    \icmlauthor{Xiaoshu Chen}{nudt}
    \icmlauthor{Ke Liang}{nudt}
    \icmlauthor{Siwei Wang}{nudt}
    \icmlauthor{Xinwang Liu}{nudt}
    \icmlauthor{Jian Huang}{nudt}
  \end{icmlauthorlist}

  \icmlaffiliation{nudt}{National University of Defense Technology, China}

  \icmlcorrespondingauthor{Sihang Zhou}{sihangjoe@gmail.com}

  \icmlkeywords{Machine Learning, Retrieval-Augmented Generation}

  \vskip 0.3in]

\printAffiliationsAndNotice{\icmlEqualContribution}
% ...




\begin{abstract}
Recent GraphRAG methods integrate graph structures into text indexing and retrieval, using knowledge graph triples to connect text chunks, thereby improving retrieval coverage and precision. However, we observe that treating text chunks as the basic unit of knowledge representation rigidly groups multiple atomic facts together, limiting the flexibility and adaptability needed to support diverse retrieval scenarios. Additionally, triple-based entity linking is sensitive to relation-extraction errors, which can lead to missing or incorrect reasoning paths and ultimately hurt retrieval accuracy. To address these issues, we propose the Atom-Entity Graph, a more precise and reliable architecture for knowledge representation and indexing. In our approach, knowledge is stored as knowledge atoms, namely individual, self-contained units of factual information, rather than coarse-grained text chunks. This allows knowledge elements to be flexibly reassembled without mutual interference, thereby enabling seamless alignment with diverse query perspectives. Edges between entities simply indicate whether a relationship exists. By combining personalized PageRank with relevance-based filtering, we maintain accurate entity connections and improve the reliability of reasoning. Theoretical analysis and experiments on five public benchmarks show that the proposed AtomicRAG algorithm outperforms strong RAG baselines in retrieval accuracy and reasoning robustness. Code: \href{https://github.com/7HHHHH/AtomicRAG}{https://github.com/7HHHHH/AtomicRAG}.
\end{abstract}

\section{Introduction}
\begin{figure*}[!htb]
    \centering
    \includegraphics[width=0.9\textwidth]{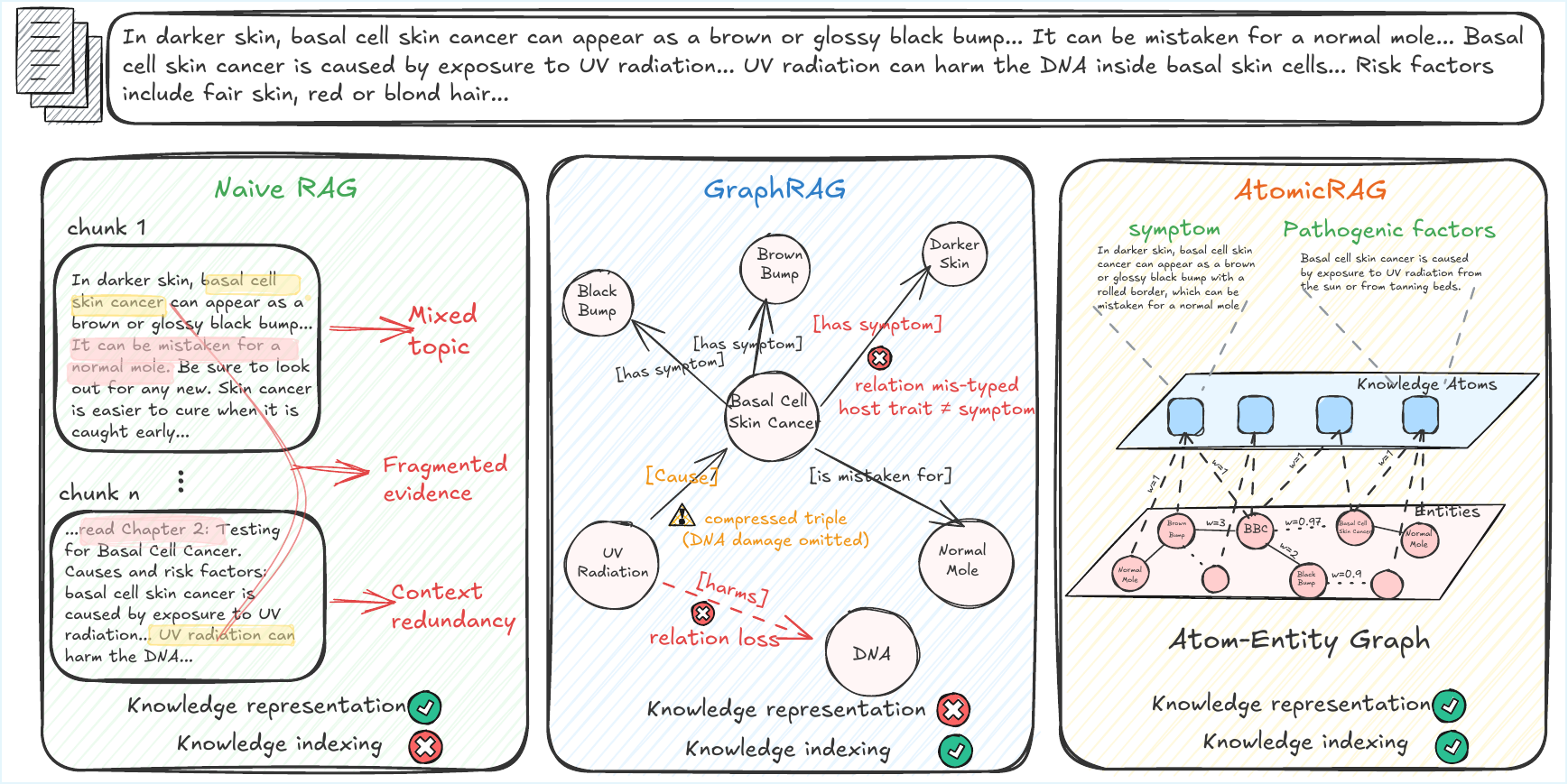}
    \caption{Comparison of knowledge representation and indexing for three classes of methods. Native RAG uses coarse text chunks as basic storage units and indexes them via semantic similarity. GraphRAG organizes knowledge with triples or chunk-level nodes, building connections through relation edges to facilitate global indexing. The proposed Atom--Entity Graph instead represents the corpus with fine-grained knowledge atoms, connects entities via co-occurrence relationships, and yields more stable and accurate connections between knowledge pieces.}

    \label{fig:m1}
\end{figure*}

Retrieval-augmented generation (RAG)~\cite{RAG2,RAG1} has become a standard paradigm to connect large language models (LLMs)~\cite{DS-R1,Qwen} to external corpora for knowledge-intensive tasks, improving factual grounding and answer accuracy. Classic RAG pipelines~\cite{chunk-RAG1,chunk-RAG2} rely on chunk-based retrieval: documents are split into fixed-length text blocks, embedded, and retrieved using a dense similarity search. This simple and efficient design largely preserves the original semantics, but treats knowledge as isolated fragments. It ignores inner relations among chunks and often introduces redundant context, which makes it brittle on queries that require integrating dispersed evidence or following multi-step reasoning chains.

Another branch of methods, GraphRAG algorithms typically adopt one of two organizational approaches to knowledge: they either replace the original corpus entirely with a triple-based graph as the principal knowledge repository~\cite{G3,G5,G1,G2, G4, logitRAG}, or link a knowledge graph with textual chunks from the corpus. However, the triple-replacement strategy unavoidably discards significant contextual information during the simplification process, which is often essential for accurate question answering. Meanwhile, the graph–chunk linking strategy—much like conventional chunk-based retrieval—constrains knowledge to fixed segments, limiting its ability to dynamically reorganize information according to varied query needs and potentially hindering precise retrieval of relevant content. Additionally, extracting reliable triple relations in open-domain settings remains a challenge. Errors in triple construction can result in incomplete or incorrect reasoning paths during retrieval, ultimately compromising the quality of generated answers. As illustrated in Fig.~\ref{fig:m1}, GraphRAG can be organized by graphs that appear well indexed yet are unreliable as knowledge representations. For instance, a host attribute is incorrectly typed as a disease symptom (e.g., linking \emph{Basal Cell Skin Cancer} to \emph{Darker Skin}), a multi-step causal explanation involving UV exposure and DNA damage is collapsed into a single coarse relation, and key mechanistic connections are missing altogether. Retrieval guided by such an index therefore follows structurally plausible but informationally distorted paths.

To address the limitations of existing methods in knowledge organization and indexing, this paper introduces AtomicRAG, a novel retrieval-augmented framework centered around an Atom–Entity Graph (AEG). During pre-processing, the corpus is decomposed into fine-grained, self-contained units called knowledge atoms, which serve as the basic representation of information. The AEG structurally organizes these atoms along with entities extracted from them, using unlabeled edges to capture co-occurrence relations—both between entities (Relevance Edges) and between atoms and their contained entities (Containment Edges). This graph-based representation enables flexible and precise retrieval, whether locally or globally, while providing a stable and reliable structure for semantic search.
At retrieval time, AtomicRAG adopts a query-decomposition strategy that decouples reasoning from retrieval. Complex queries are adaptively broken down into atom-aligned sub-questions, enabling fine-grained matching with the knowledge base. A entity-resonance graph retrieval mechanism then combines semantic similarity and graph-based relevance propagation to identify the most pertinent atoms. Finally, a filtering step removes redundant or irrelevant content, ensuring that only concise, high-utility evidence collected by all sub-questions is passed to the language model. This setting not only reduces noise during retrieval but also enhances the factual accuracy and provenance transparency of the generated answers.

The contributions of this paper are threefold: (1) We propose the Atom–Entity Graph (AEG), a novel knowledge representation that is more flexible and robust than conventional chunk-based or relation-labeled graphs. (2) We design a query-adaptive retrieval pipeline that first decomposes complex questions into atomic sub-questions and then uses entity-resonance graph propagation to accurately gather concise and relevant evidence. (3) Through both theoretical analysis and extensive experiments on five benchmarks, we demonstrate that AtomicRAG outperforms strong baselines in retrieval accuracy and reasoning robustness, especially for multi-hop queries that require evidence composition.

\section{Related Work}
\subsection{Retrieval-Augmented Generation}
Retrieval-augmented generation (RAG) grounds LLM outputs by retrieving external evidence and conditioning generation on the retrieved context~\cite{memorag,Hou}. 
Beyond passage-level indexing, Dense X Retrieval shows that proposition-level retrieval can improve retrieval quality and downstream QA under a fixed compute budget~\cite{DenseX}. 
To reduce ambiguity when retrieved passages are detached from their original document context, Contextual Retrieval augments each chunk with automatically generated, chunk-specific context for both dense and BM25 retrieval, and further benefits from reranking. 
For multi-step information needs, HyDE enriches queries via hypothetical document embeddings~\cite{HyDE}, while IRCoT, Iter-RetGen, and multi-hop dense retrievers such as MDR interleave or iterate retrieval with intermediate reasoning signals to progressively locate supporting evidence~\cite{IRCOT,Iter,MDR}. Complementary efforts such as RAPTOR and LLMLingua improve long-context usability via hierarchical organization and prompt compression, and Self-RAG studies on-demand retrieval with self-critique during decoding~\cite{RAPTOR,LLMLingua,selfrag}. 
Despite these advances, composing stable cross-document evidence chains remains challenging in multi-hop settings.

\begin{figure*}[!htb]
    \centering
    \includegraphics[width=0.97\textwidth]{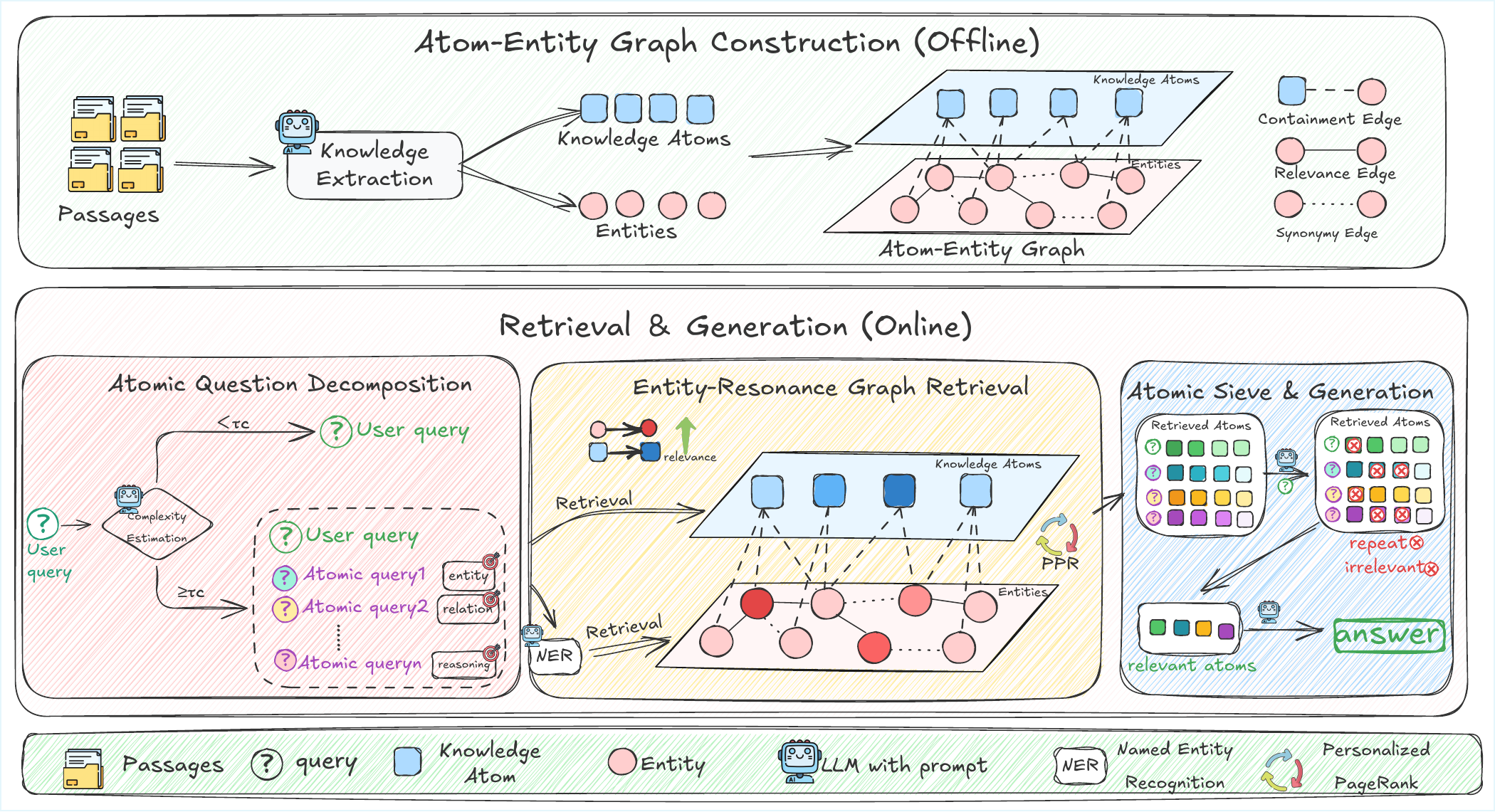}
    \caption{Overview of AtomicRAG. During the preprocessing phase, we construct an unlabeled Atom–Entity Graph (AEG) that atomizes the corpus into minimal knowledge atoms linked via entities and co-occurrence relationships. Specifically, as illustrated in the figure, our co-occurrence relationships fall into three types: containment, relevance, and synonymy. At retrieval time, a complex query is optionally decomposed into atomic sub-queries, which seed entity-resonance propagation over the AEG to retrieve multi-hop evidence. A final atomic sieve filters and merges retrieved atoms into a compact, deduplicated context for grounded answer generation.}
    \label{fig:atomicrag-structure}
\end{figure*}
\subsection{Graph-based RAG}
GraphRAG~\cite{HyperGraphRAG,routeRAG,graph-R1} organizes evidence units and entity associations with graph structures, extending retrieval from similarity-based top-k chunks to composable subgraph or path-level evidence. Microsoft’s GraphRAG~\cite{graphRAG} induces graph structure with large language models and leverages community-level summaries to strengthen cross-document aggregation and query-focused synthesis. The HippoRAG~\cite{Hippo1,Hippo2} line combines knowledge graphs with personalized PageRank, propagating from query seeds on the graph to integrate multi-hop information in a single retrieval process. LightRAG~\cite{Light} introduces graph structures and a two-stage retrieval pipeline to balance coverage and efficiency. GFM-RAG~\cite{GFM-RAG} further employs graph neural networks to enhance multi-hop reasoning over structured knowledge. Overall, the gains of graph-based approaches often hinge on high-quality graph construction: entity coverage, edge correctness and consistency, and the cost of continuous updates directly affect the reliability of graph traversal and path composition; when graphs are noisy or incomplete, multi-hop propagation and stitching can amplify errors and produce unstable evidence chains.
\section{Method}
Figure~\ref{fig:atomicrag-structure} gives an overview of AtomicRAG, which operates in four stages.
(i) Atom--Entity Graph Construction (Offline), which builds a persistent Atom--Entity Graph (AEG) as the knowledge store;
(ii) Atomic Question Decomposition, which rewrites $q$ into atom-level sub-questions;
(iii) Entity-Resonance Graph Retrieval, which propagates signals over the AEG to select candidate atoms; and
(iv) Atomic Sieve, which filters and orders atoms to form a compact context for downstream generation.
Together, these stages form an end-to-end pipeline for composing multi-hop evidence to answer complex queries.

\subsection{Atom--Entity Graph Construction}

\textbf{Joint extraction of knowledge atoms and triples.}
Given a corpus $\mathcal{C}=\{d_i\}_{i=1}^N$, we first apply an instruction-tuned LLM with an entity-centric prompt to each document $d_i$ and obtain canonical entities $\mathcal{E}_i$.
We then feed $(d_i,\mathcal{E}_i)$ into a second prompt for \emph{joint extraction of knowledge atoms and triples}, producing atoms $\mathcal{A}_i=\{a_{i,1},\ldots,a_{i,n_i}\}$ and triples $\mathcal{T}_i\subseteq \mathcal{E}_i\times \mathcal{R}\times \mathcal{E}_i$, where $\mathcal{R}$ denotes textual relation labels.
Triples are used \emph{only} to derive auxiliary entity–entity edges in the graph; they are never retrieved directly as evidence units. The AEG contains two node types.

\emph{Knowledge atom nodes.}
A knowledge atom $a_{i,j}\in\mathcal{A}_i$ is a minimal, self-contained natural-language statement parsed from $d_i$.
It is written to be context-complete and non-anaphoric, i.e., it avoids unresolved pronouns (e.g., ``it'', ``they'') and remains interpretable in isolation.
In practice, an atom typically focuses on a single topic or fact (e.g., the symptoms of basal cell skin cancer), rather than entangling multiple facets.
Each atom is retrievable as an independent evidence unit and is annotated with the entities it involves, denoted by $\mathcal{E}(a_{i,j})\subseteq \mathcal{E}_i$.

\emph{Entity nodes.}
Entities $\mathcal{E}_i=\{e_{i,1},\ldots,e_{i,m_i}\}$ are canonical concepts aggregated from surface mentions in $d_i$ (e.g., aliases or coreferent mentions).

Aggregating over the corpus yields global inventories
$\mathcal{A}=\bigcup_{i=1}^N\mathcal{A}_i$,
$\mathcal{E}=\bigcup_{i=1}^N\mathcal{E}_i$, and
$\mathcal{T}=\bigcup_{i=1}^N\mathcal{T}_i$.

\textbf{Connectivity organization.}
We define the Atom--Entity Graph as a heterogeneous, weighted graph $G=(V,\mathcal{L})$ with two node types (atoms and entities) and three edge types. We deliberately omit textual predicate labels and instead attach scalar weights to edges.
Edges in $\mathcal{L}$ are treated as bidirectional during propagation.
We organize connectivity using three co-occurrence relationships: \emph{containment edges}, \emph{relevance edges}, and \emph{synonym edges}.

\emph{Containment edges (weight $1$).}
We connect each atom to the entities it mentions, with unit weight:
\begin{equation}
\mathcal{L}_{\mathrm{cont}}
=\{(a,e)\mid a\in\mathcal{A},\, e\in\mathcal{E}(a)\},\qquad
w(a,e)=1.
\end{equation}

\emph{Relevance edges (weight = \#distinct relation types).}
To strengthen cross-sentence and cross-document connectivity without committing to brittle predicate semantics, we derive entity--entity relevance edges from triples.
For each entity pair $(e,e')$, we assign a scalar relevance weight equal to the number of \emph{distinct} relation labels connecting them:
\begin{equation}
w(e,e')
=\Bigl|\{\, r \mid (e,r,e')\in\mathcal{T}\ \text{or}\ (e',r,e)\in\mathcal{T}\,\}\Bigr|.
\end{equation}
Whenever $w(e,e')>0$, we add an undirected relevance edge $(e,e')$.
We do \emph{not} store textual labels $r$ on the graph; only the scalar weight $w(e,e')$ is retained.

\emph{Synonym edges (weight = similarity).}
To alleviate fragmentation caused by aliases or near-synonymous forms, we connect entities with similar representations.
Let $\mathbf{z}_e$ denote the embedding of entity $e$ produced by our shared encoder (detailed in the next paragraph).
If $\cos(\mathbf{z}_e,\mathbf{z}_{e'})\ge \tau_s$, we add a synonym edge $(e,e')$ with weight
$w(e,e')=\cos(\mathbf{z}_e,\mathbf{z}_{e'})$:
\begin{equation}
\mathcal{L}_{\mathrm{syn}}
=\{(e,e') \mid e,e'\in\mathcal{E},\ \cos(\mathbf{z}_e,\mathbf{z}_{e'})\ge \tau_s\}.
\end{equation}

The final edge set is $\mathcal{L}=\mathcal{L}_{\mathrm{cont}}\cup\mathcal{L}_{\mathrm{rel}}\cup\mathcal{L}_{\mathrm{syn}}$, where $\mathcal{L}_{\mathrm{rel}}=\{(e,e')\mid e,e'\in\mathcal{E},\, w(e,e')>0\}$.

\textbf{Vector representation storage.}
Atoms, entities, and (sub-)queries are embedded into a shared vector space using a common encoder $f_\theta(\cdot)$, i.e.,
$\mathbf{z}_a=f_\theta(a)$, $\mathbf{z}_e=f_\theta(e)$, and $\mathbf{z}_{q'}=f_\theta(q')$.
In practice, we store (i) an approximate nearest neighbor index over atom embeddings for semantic retrieval, (ii) an embedding table for entities, and (iii) the sparse weighted adjacency induced by $\mathcal{L}$ for graph propagation.
Thus, each knowledge-atom node is simultaneously a minimal semantic unit and a structural hook into the AEG.
\textbf{Proposition 1.} \textit{The Atom--Entity Graph provides a more comprehensive and more robust knowledge representation.}\vspace{-2mm}
\begin{proof}
We provide experimental evidence in Section~\ref{exp:graph} and a formal proof in Appendix~\ref{proof1}.
\end{proof}

\subsection{Atomic Question Decomposition}
Complex queries often implicitly decompose into multiple sub-questions whose answers rely on distinct evidence fragments.
Treating such queries as a single retrieval unit forces the retriever to match an entangled signal, thereby amplifying semantic drift and retrieval noise.
To mitigate this mismatch, we optionally perform \emph{atomic question decomposition}, producing query units whose granularity better matches that of our atomic knowledge.
Through this decomposition, atomic sub-queries and knowledge atoms are aligned at the level of \emph{evidence demand}: each $q^{(t)}$ is designed to seek a single, self-contained atomic evidence unit, reducing cross-facet entanglement during retrieval.

Given a query $q$, we prompt an LLM with a rubric-style instruction to assign a structural complexity score $c(q)\in[0,10]$; details of the scoring prompt are provided in the appendix.
If $c(q)$ exceeds a fixed threshold $\tau_c$, the query is decomposed into a small set of atomic sub-queries $\{q^{(1)},\dots,q^{(m)}\}$ with $m\le m_{\max}$.
We explicitly instruct the LLM to generate sub-queries that each target a specific facet of $q$ (e.g., grounding an entity mention, specifying a relation, or resolving an intermediate reasoning step), rather than producing arbitrary paraphrases.
We then define the effective query set
\begin{equation}
\widetilde{\mathcal{Q}}(q)=
\begin{cases}
\{q\}\cup\{q^{(1)},\dots,q^{(m)}\}, & c(q)\ge\tau_c,\\
\{q\}, & \text{otherwise}.
\end{cases}
\end{equation}
Each $q'\in\widetilde{\mathcal{Q}}(q)$ is processed independently in the subsequent retrieval stage, which reduces early entanglement between distinct evidence requirements.

\textbf{Proposition 2.} \textit{Granularity alignment between queries and atomic knowledge improves retrieval efficiency.}\vspace{-2mm}
\begin{proof}
We provide a formal proof in Appendix~\ref{proof2}.
\end{proof}

\subsection{Entity-Resonance Graph Retrieval}
Pure dense retrieval over atoms lacks mechanisms for organizing multi-hop evidence, whereas explicit symbolic reasoning over noisy predicate-typed relations is brittle.
Entity-Resonance Graph Retrieval instead uses the AEG as an unlabeled scaffold: it softly propagates query signals through shared entities, inducing interpretable evidence chains without relying on semantic predicates.

\textbf{Query-specific personalization.}
For each effective query $q' \in \widetilde{\mathcal{Q}}(q)$, we initialize a personalization distribution over graph nodes by combining two signals:
(i) dense similarity between $q'$ and atomic embeddings obtained from the atom index; and
(ii) entity mentions extracted from $q'$ and mapped to entity nodes.
Let $r^{(0)}_{\text{atom}}(q',v)$ and $r^{(0)}_{\text{ent}}(q',v)$ be non-negative seed weights on $V$,
with $r^{(0)}_{\text{atom}}(q',v)=0$ for $v\notin\mathcal{A}$ and $r^{(0)}_{\text{ent}}(q',v)=0$ for $v\notin\mathcal{E}$.
We attenuate the atomic seeds by a scalar $\alpha$ and then normalize the combined scores:
\begin{equation}
\tilde{\pi}_{q'}(v)= \alpha\, r^{(0)}_{\text{atom}}(q',v) + r^{(0)}_{\text{ent}}(q',v),
\end{equation}
\begin{equation}
\pi_{q'}(v)=\frac{\tilde{\pi}_{q'}(v)}{\sum_{u\in V}\tilde{\pi}_{q'}(u)},
\qquad
\sum_{v\in V}\pi_{q'}(v)=1.
\end{equation}
Here $\alpha$ down-weights direct atom-level similarity relative to entity-based signals; in all experiments, we set $\alpha = 0.1$, which biases the initialization toward entities while retaining a small amount of atomic evidence.

\textbf{Resonance propagation.}
Let $P$ be the row-normalized transition matrix of $G$.
We compute a personalized PageRank vector $\mathbf{r}_{q'}$ as the fixed point of
\begin{equation}
\mathbf{r}_{q'} = \rho\,\boldsymbol{\pi}_{q'} + (1-\rho) P^\top \mathbf{r}_{q'},
\end{equation}
where $\boldsymbol{\pi}_{q'}$ is the vector form of $\pi_{q'}$ and $\rho\in(0,1)$ is the restart probability.
We set $\rho = 0.3$ throughout.
This propagation distributes relevance mass along atom--entity--atom paths, amplifying atoms that are structurally well supported by the entities mentioned (or resolved via auxiliary links) in $q'$.
Atomic relevance scores are given directly by
\begin{equation}
s_{q'}(a) = r_{q'}(a), \qquad a\in\mathcal{A},
\end{equation}
and high-mass paths in $G$ constitute \emph{entity-resonance chains}, providing an explicit account of evidence flow for $q'$.

\subsection{Atomic Sieve and Grounded Generation}
Graph-based propagation over the AEG can still surface loosely related or redundant atoms.
We therefore apply a final semantic filtering step at the atomic level to ensure precision without reverting to coarse retrieval units.

\textbf{Atomic filtering.}
For each effective query $q' \in \widetilde{\mathcal{Q}}(q)$, we first select a small candidate set of atoms by their resonance scores:
\begin{equation}
\mathcal{R}(q') = \operatorname*{TopK}_{a\in\mathcal{A}} s_{q'}(a),
\qquad |\mathcal{R}(q')| = K,
\end{equation}
with $K = 25$ in all experiments.
Candidates from the original query and all sub-queries are merged as
\begin{equation}
\mathcal{R}(q) = \bigcup_{q'\in\widetilde{\mathcal{Q}}(q)} \mathcal{R}(q').
\end{equation}
We then obtain a filtered subset $\mathcal{S}(q) \subseteq \mathcal{R}(q)$ by prompting an instruction-tuned LLM to judge, for each $a\in\mathcal{R}(q)$, whether $a$ is necessary and relevant to the original query $q$.
Thus, sub-queries $q'$ are only used to expose diverse candidates, while all inclusion decisions are grounded in the original information need.

\textbf{Aggregation and generation.}
Filtered atoms are further merged at the source-document level to form the final evidence set
\begin{equation}
\mathcal{A}^*(q) \subseteq \mathcal{S}(q),
\end{equation}
where atoms that refer to overlapping spans from the same document are combined into a single citation unit to avoid redundant text and keep the context compact.

\begin{table*}[t]
\centering
\footnotesize
\caption{Performance comparison on Graph-Bench and multi-hop QA benchmarks. Fact, Reason, Summ., and Creat. denote Fact Retrieval, Complex Reasoning, Contextual Summarization, and Creative Generation, respectively. The final Avg. is the mean across all tasks. Best results are in \textbf{bold} and second-best results are \underline{underlined}. The improvement row reports absolute score gains (in points) of Ours over the best baseline; $\uparrow$ denotes increases.}
\setlength{\tabcolsep}{3.0pt}
\renewcommand{\arraystretch}{1.18}
\resizebox{\textwidth}{!}{%
\begin{tabular}{lccccccccccccccc}
\toprule
\multirow[c]{2}{*}{\parbox[c]{2.3cm}{\centering\textbf{Method}}} &
\multicolumn{5}{c}{Graph-Bench (Medical)} &
\multicolumn{5}{c}{Graph-Bench (Novel)} &
\multicolumn{4}{c}{Multi-hop QA} &
\multirow[c]{2}{*}{\parbox[c]{1.2cm}{\centering\textbf{Avg.}}} \\
\cmidrule(lr){2-6}\cmidrule(lr){7-11}\cmidrule(lr){12-15}
& Fact & Reason & Summ. & Creat. & Avg. &
  Fact & Reason & Summ. & Creat. & Avg. &
  HotpotQA & 2Wiki & MuSiQue & Avg. & \\
\midrule

\multicolumn{16}{c}{\textbf{Vanilla Retrieval-Augmented Generation}}\\
\specialrule{0.4pt}{0.3ex}{0.6ex}
RAG (w/o rerank) & 63.7 & 57.6 & 63.7 & 58.9 & 61.0 & 58.7 & 41.4 & 50.1 & 41.5 & 47.9 & 54.0 & 31.8 & 32.2 & 39.3 & 50.3 \\
RAG (w rerank)   & 64.7 & 58.6 & 65.8 & 60.6 & 62.4 & \underline{60.9} & 42.9 & 51.3 & 38.3 & 48.3 & 54.5 & 32.1 & 33.3 & 40.0 & 51.2 \\

\addlinespace[2pt]
\specialrule{0.6pt}{0.4ex}{0.6ex}
\multicolumn{16}{c}{\textbf{Graph-based Retrieval-Augmented Generation}}\\
\specialrule{0.4pt}{0.3ex}{0.6ex}
MS-GraphRAG (local)  & 38.6 & 47.0 & 41.8 & 53.1 & 45.1 & 49.3 & \underline{50.9} & 64.4 & 39.1 & 50.9 & 52.2 & 37.4 & 37.8 & 42.5 & 46.5 \\
MS-GraphRAG (global) & 16.4 & 15.6 & 19.8 & 20.8 & 18.1 & 36.9 & 43.2 & 56.9 & 41.1 & 44.5 & 37.6 & 39.5 & 35.4 & 37.5 & 33.0 \\
HippoRAG            & 56.1 & 55.8 & 59.8 & 64.4 & 59.0 & 52.9 & 38.5 & 48.7 & 38.9 & 44.8 & 38.1 & 29.8 & 25.7 & 31.2 & 46.2 \\
HippoRAG2           & \underline{66.3} & 61.9 & 63.1 & \underline{68.1} & \underline{64.8} &
                      60.1 & \textbf{53.4} & 64.1 & 48.3 & \underline{56.5} &
                      67.5 & 47.5 & 41.5 & 52.2 & \underline{58.3} \\
LightRAG            & 63.3 & 61.3 & 63.1 & 67.9 & 63.9 & 58.6 & 49.1 & 48.9 & 23.8 & 45.1 & 57.0 & 44.2 & 23.5 & 41.6 & 51.0 \\
Fast-GraphRAG       & 60.9 & 61.7 & \underline{67.9} & 65.9 & 64.1 & 56.9 & 48.5 & 56.4 & 46.2 & 52.0 & 62.3 & 47.8 & 42.1 & 50.7 & 56.1 \\
RAPTOR              & 54.0 & 53.2 & 58.7 & 62.4 & 57.1 & 49.3 & 38.6 & 47.1 & 38.0 & 43.2 & 66.4 & 50.1 & \underline{43.6} & \underline{53.4} & 51.0 \\
Lazy-GraphRAG       & 60.3 & 47.8 & 57.3 & 62.2 & 56.9 & 51.7 & 49.2 & 58.3 & 43.2 & 50.6 & 54.6 & 40.4 & 39.1 & 44.7 & 51.3 \\
KGP                 & 52.3 & 51.5 & 54.5 & 63.8 & 55.5 & 54.2 & 46.3 & 51.2 & 40.3 & 48.0 & 52.3 & 38.6 & 37.2 & 42.7 & 49.3 \\
StructRAG           & 55.4 & 56.2 & 62.5 & 60.2 & 58.6 & 53.8 & 46.3 & 54.3 & 42.2 & 49.1 & 49.2 & \underline{52.3} & 24.9 & 42.1 & 50.7 \\
KET-RAG             & 60.4 & 39.6 & 45.3 & 43.0 & 47.1 & 55.4 & 36.6 & 52.5 & 46.0 & 47.6 & 45.6 & 24.5 & 22.6 & 30.9 & 42.9 \\
GFM-RAG             & 63.5 & \underline{67.3} & 51.5 & 63.3 & 61.4 & 51.0 & 49.8 & \underline{66.3} & \underline{57.8} & 56.2 & \underline{68.6} & 43.6 & 39.9 & 50.7 & 56.6 \\

\specialrule{0.6pt}{0.4ex}{0.2ex}
\rowcolor{black!7}
Ours & \textbf{72.6} & \textbf{74.8} & \textbf{76.8} & \textbf{68.3} & \textbf{73.1} &
       \textbf{61.0} & \textbf{53.4} & \textbf{68.5} & \textbf{60.0} & \textbf{60.7} &
       \textbf{70.5} & \textbf{56.8} & \textbf{50.9} & \textbf{59.4} & \textbf{64.9} \\
\rowcolor{black!7}
\textit{Improv. vs best baseline} &
\textit{$\uparrow$6.3} & \textit{$\uparrow$7.5} & \textit{$\uparrow$8.9} & \textit{$\uparrow$0.2} & \textit{$\uparrow$8.3} &
\textit{$\uparrow$0.1} & \textit{-} & \textit{$\uparrow$2.2} & \textit{$\uparrow$2.2} & \textit{$\uparrow$4.3} &
\textit{$\uparrow$1.9} & \textit{$\uparrow$4.5} & \textit{$\uparrow$7.3} & \textit{$\uparrow$6.0} & \textit{$\uparrow$6.6} \\
\bottomrule
\end{tabular}%
}

\label{tab:graphbench_main}
\end{table*}
\section{Experiments}
This section presents the experimental setup, main results, and analyses. We answer the following research questions (RQs): \textbf{RQ1:} Does AtomicRAG outperform existing methods? \textbf{RQ2:} How does each major component of AtomicRAG contribute to performance? \textbf{RQ3:} Is our Atom–Entity Graph better than alternative graph organizations? \textbf{RQ4:} Can the Entity-Resonance Graph Retrieval strategy improve retrieval accuracy and efficiency? \textbf{RQ5:} What are the costs of AtomicRAG in indexing and generation? Additional analyses are provided in the appendix.
\subsection{Experimental Setup}

\textbf{Datasets and Metrics.}
We evaluate the effectiveness of AtomicRAG on two domain-specific benchmarks from Graph-Bench~\cite{G-ben} and three widely used multi-hop QA datasets (HotpotQA~\cite{hotpotqa}, 2WikiMultiHopQA~\cite{2wiki}, and MuSiQue~\cite{musique}).
For Graph-Bench \textit{Medical} and \textit{Novel}, queries are categorized into four question types with increasing difficulty: \textit{Fact Retrieval}, \textit{Complex Reasoning}, \textit{Contextual Summarization}, and \textit{Creative Generation}.
For all five datasets, we follow the Graph-Bench preprocessing protocol for consistency: 
documents are segmented into chunks of 256 tokens with an overlap of 32 tokens.
As the evaluation metric, we adopt the \textit{Answer Accuracy (ACC)} proposed by Graph-Bench, 
which combines LLM-based judging with embedding-based semantic matching; 
detailed definitions and implementation are provided in Appendix~\ref{app:baseline_chunk}. 

\textbf{Baselines and Implementation Details.}
We group baselines into two categories.
(i) Vanilla RAG: a standard dense-retrieval pipeline with the same generator, 
evaluated both without reranking and with reranking.
(ii) Graph-enhanced RAG: representative systems that organize evidence with explicit structures, including MS-GraphRAG~\cite{graphRAG}, RAPTOR~\cite{RAPTOR}, LightRAG~\cite{Light}, HippoRAG~\cite{Hippo1}, HippoRAG2~\cite{Hippo2}, Fast-GraphRAG~\cite{fastgraphrag}, LazyGraphRAG~\cite{LazyGraphRAG}, KET-RAG~\cite{KETrag}, KGP~\cite{KGP}, StructRAG~\cite{StructRAG}, and GFM-RAG~\cite{GFM-RAG}.  
To ensure a controlled comparison, all methods use the same embedding model (\texttt{BAAI/bge-large-en-v1.5}).
For both answer generation and LLM-based evaluation, we use the same backbone LLM (\texttt{GPT-4o-mini}).
Full specifications are in Appendix~\ref{app:metrics}.

\subsection{Main Results (RQ1)}
\textbf{Overall comparison.}
To assess the effectiveness of our method, we compare it with strong vanilla RAG variants and a broad set of graph-enhanced RAG baselines across Graph-Bench and multi-hop QA benchmarks. Results are reported in Table~\ref{tab:graphbench_main}. Our method achieves the best overall average across all task columns (Avg.=64.9), consistently outperforming all baselines. Notably, we attain the top score on most task columns and tie for the best on Graph-Bench (Novel) \textit{Reason}, indicating that the improvement is not driven by a single dataset or question type, but holds broadly across tasks.

\textbf{Across benchmarks and domains.}
Performance gains remain consistent across all benchmark groups. On Graph-Bench (Medical), we reach 73.1 Avg., improving over the best baseline by +8.3; on Graph-Bench (Novel), we achieve 60.7 Avg.\ with a +4.3 gain; on Multi-hop QA, we obtain 59.4 Avg.\ with a +6.0 gain. The largest margins appear on harder benchmarks and domains that require chaining dispersed evidence (e.g., MuSiQue), a pattern consistent with improved multi-hop evidence composition rather than gains limited to single-hop matching.

\textbf{Across question types.}
Our method improves performance uniformly across question types, reflecting broad coverage rather than isolated, type-specific gains. On Graph-Bench (Medical), we simultaneously improve \textit{Fact/Reason/Summ.} by +6.3/+7.5/+8.9, suggesting that the same design choices strengthen factual grounding, multi-step evidence chaining, and cross-atom synthesis instead of over-optimizing a single skill. On Graph-Bench (Novel), we stay competitive (often best or near-best) on already-strong types while still lifting the weaker ones, which narrows the gap between categories and raises the overall ceiling. In contrast, several baselines show higher variance across question types—excelling on a subset but degrading on others—whereas our method remains consistently strong across all categories, indicating better robustness to question-style shifts.

\subsection{Ablation Results (RQ2)}
\begin{table}[t]
\centering
\footnotesize
\caption{Ablation on three datasets. Parentheses indicate score drops relative to AtomicRAG.
ERGR: Entity-Resonance Graph Retrieval; AQD: Atomic Question Decomposition; AS: Atomic Sieve; KA: Knowledge Atomization.}
\setlength{\tabcolsep}{5.2pt}
\renewcommand{\arraystretch}{1.14}
\resizebox{\linewidth}{!}{%
\begin{tabular}{lcccc}
\toprule
\textbf{Method} & \textbf{HotpotQA} & \textbf{Medical} & \textbf{Novel} & \textbf{Avg.} \\
\midrule
AtomicRAG & 70.5 & 73.2 & 60.7 & 68.1 \\
\midrule

\multicolumn{5}{l}{\textcolor{gray}{\footnotesize\textbf{Single-module removal}}} \\
\cmidrule(lr){1-5}
w/o ERGR  &
\underline{68.4}\textcolor{gray}{\scriptsize($\downarrow$2.1)} &
72.0\textcolor{gray}{\scriptsize($\downarrow$1.2)} &
59.1\textcolor{gray}{\scriptsize($\downarrow$1.6)} &
66.5\textcolor{gray}{\scriptsize($\downarrow$1.6)} \\
w/o AQD   &
67.8\textcolor{gray}{\scriptsize($\downarrow$2.7)} &
\underline{72.3}\textcolor{gray}{\scriptsize($\downarrow$0.9)} &
\underline{60.0}\textcolor{gray}{\scriptsize($\downarrow$0.7)} &
66.7\textcolor{gray}{\scriptsize($\downarrow$1.4)} \\
w/o AS    &
68.0\textcolor{gray}{\scriptsize($\downarrow$2.5)} &
71.5\textcolor{gray}{\scriptsize($\downarrow$1.7)} &
59.8\textcolor{gray}{\scriptsize($\downarrow$0.9)} &
66.4\textcolor{gray}{\scriptsize($\downarrow$1.7)} \\
w/o KA    &
62.2\textcolor{gray}{\scriptsize($\downarrow$8.3)} &
64.3\textcolor{gray}{\scriptsize($\downarrow$8.9)} &
52.7\textcolor{gray}{\scriptsize($\downarrow$8.0)} &
59.7\textcolor{gray}{\scriptsize($\downarrow$8.4)} \\
\midrule

\multicolumn{5}{l}{\textcolor{gray}{\footnotesize\textbf{Combined removal}}} \\
\cmidrule(lr){1-5}
w/o \{ERGR, AQD\} &
66.0\textcolor{gray}{\scriptsize($\downarrow$4.5)} &
70.8\textcolor{gray}{\scriptsize($\downarrow$2.4)} &
58.6\textcolor{gray}{\scriptsize($\downarrow$2.1)} &
65.1\textcolor{gray}{\scriptsize($\downarrow$3.0)} \\
w/o \{ERGR, AQD, AS\} &
64.5\textcolor{gray}{\scriptsize($\downarrow$6.0)} &
69.5\textcolor{gray}{\scriptsize($\downarrow$3.7)} &
57.5\textcolor{gray}{\scriptsize($\downarrow$3.2)} &
63.8\textcolor{gray}{\scriptsize($\downarrow$4.3)} \\
w/o \{ERGR, AQD, AS, KA\} &
54.0\textcolor{gray}{\scriptsize($\downarrow$16.5)} &
61.0\textcolor{gray}{\scriptsize($\downarrow$12.2)} &
47.9\textcolor{gray}{\scriptsize($\downarrow$12.8)} &
54.3\textcolor{gray}{\scriptsize($\downarrow$13.8)} \\
\bottomrule
\end{tabular}%
}
\vspace{1pt}
\label{tab:ablation_three}
\end{table}
Table~\ref{tab:ablation_three} reports ablations on HotpotQA, Graph-Bench (Medical), and Graph-Bench (Novel), using Avg.\ as the main metric. AtomicRAG achieves 68.1 Avg. Removing any single module consistently reduces performance, confirming that the gains come from complementary components rather than any single component.

\textbf{Single-module impact.}
Each module provides a measurable benefit: removing ERGR/AQD/AS reduces Avg.\ by 1.6/1.4/1.7, respectively, while removing KA causes the largest single drop to 59.7 ($-8.4$). The effects align with their roles: AQD is most critical on HotpotQA ($-2.7$), ERGR degrades performance uniformly across datasets, and removing AS consistently hurts performance, indicating that the sieve effectively filters noisy atoms and sharpens evidence precision.

\textbf{Synergy under combined removal.}
The drops compound when modules are removed jointly: removing \{ERGR, AQD\} reduces Avg.\ by 3.0, and further removing AS increases the drop to 4.3. This demonstrates clear complementarity between decomposition (AQD), graph retrieval (ERGR), and final filtering (AS).

\textbf{Knowledge atomization is foundational.}
Removing KA leads to the largest performance degradation. In the \textit{w/o KA} variant, we disable knowledge atomization and replace atomic knowledge units with the original text chunks as retrieval units. Even with only KA removed, Avg.\ drops sharply to 59.7 ($-8.4$), and removing all modules including KA further degrades to 54.3 ($-13.8$). These results indicate that atom-level granularity is essential: without it, ERGR and AS cannot reliably form and refine evidence paths, and the system largely degenerates to coarse-grained chunk retrieval.

\begin{table}[t]
\centering
\footnotesize
\caption{\textbf{Graph structure statistics.} We report the number of nodes and edges, average degree, and average clustering coefficient for the constructed graphs on Graph-Bench (Medical).}
\label{tab:graph_stats}
\resizebox{\linewidth}{!}{%
\begin{tabular}{ccccc}
\toprule
\textbf{Method} & \textbf{\#Nodes} & \textbf{\#Edges} & \textbf{Avg.\ degree} & \textbf{Avg.\ clustering} \\
\midrule
KET-RAG   & 3{,}134  & 2{,}421   & 1.55  & 0.24 \\
LightRAG  & 1{,}942  & 2{,}220   & 2.29  & 0.09 \\
HippoRAG2 & \underline{10{,}660} & \underline{119{,}489} & \underline{22.42} & \underline{0.38} \\
GFM-RAG   & 9{,}569  & 40{,}444  & 8.45  & 0.29 \\
\rowcolor{gray!15}
Ours      & \textbf{14{,}586} & \textbf{146{,}115} & \textbf{23.04} & \textbf{0.39} \\
\bottomrule
\end{tabular}%
}
\end{table}
\begin{figure}[!htb]
    \centering
    \includegraphics[width=0.36\textwidth]{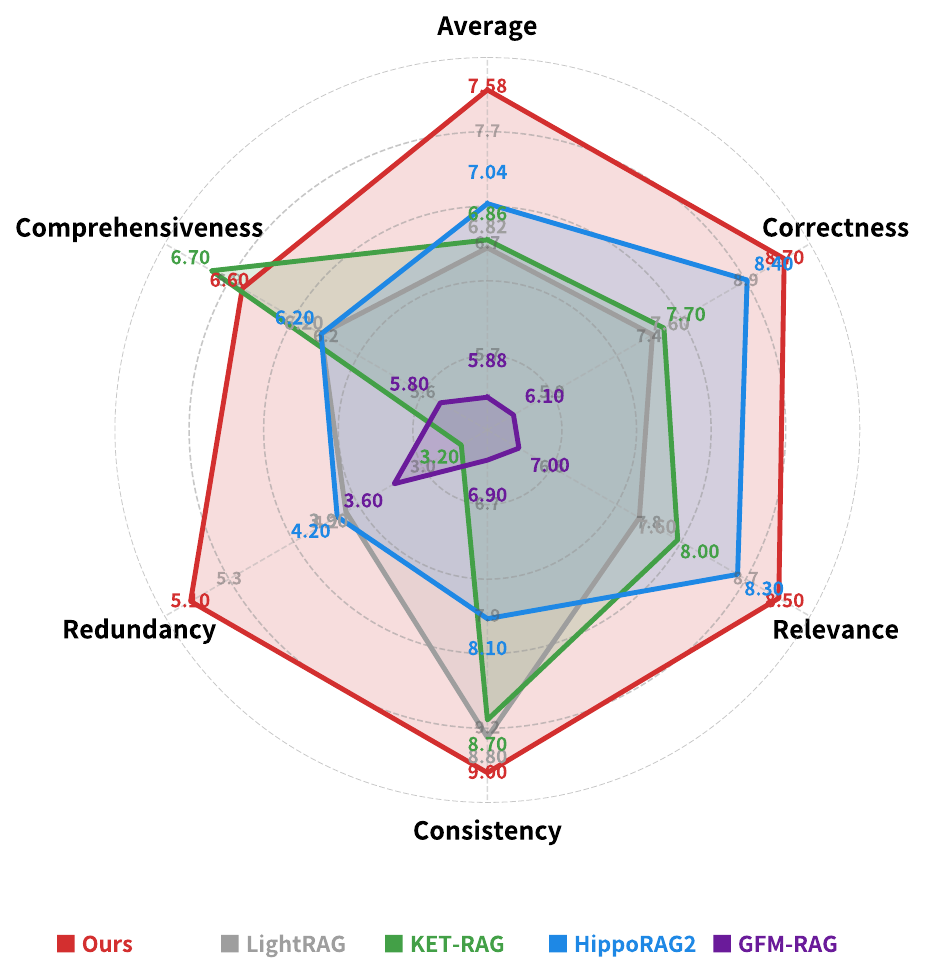}
\caption{\textbf{Semantic utility.} LLM-based assessment of 1-hop graph neighborhoods on Graph-Bench (Medical) with respect to correctness, relevance, consistency, redundancy, and comprehensiveness.}

    \label{fig:graph}
\end{figure}

\subsection{Graph Quality Analysis (RQ3)}\label{exp:graph}
We next examine the quality of the constructed graphs from both \emph{structural connectivity} and \emph{semantic utility}.

\textbf{Structural connectivity.}
Table~\ref{tab:graph_stats} reports basic structural statistics.
Compared with prior baselines, our graph is larger and exhibits slightly stronger local connectivity.
Such connectivity is desirable for composing multi-hop evidence, but it does not by itself guarantee that neighborhoods are correct or useful for grounding.

\textbf{Semantic utility.}
To assess whether neighborhoods are semantically helpful, we conduct an LLM-based neighborhood evaluation.
On Graph-Bench (Medical), we sample 10 high-frequency entities (by corpus mention count) and extract their 1-hop neighborhoods from each constructed graph.
Given the center entity and its neighborhood, we prompt \texttt{gpt-oss-120b}~\cite{GPT-oss} to score the neighborhood \emph{as a whole} along five criteria: \textit{Correctness}, \textit{Relevance}, \textit{Consistency}, \textit{Redundancy}, and \textit{Comprehensiveness}.
We average scores across entities and visualize the results in Figure~\ref{fig:graph}.
We observe a slight decrease in comprehensiveness compared with KET-RAG (6.60 vs.\ 6.70), while KET-RAG scores notably lower on redundancy.
Overall, AtomicRAG achieves higher correctness, relevance, and consistency while reducing redundancy, yielding more usable grounding neighborhoods for RAG.



\subsection{Retrieval Efficiency Analysis (RQ4)}
\begin{figure}[!htb]
    \centering
    \includegraphics[width=0.45\textwidth]{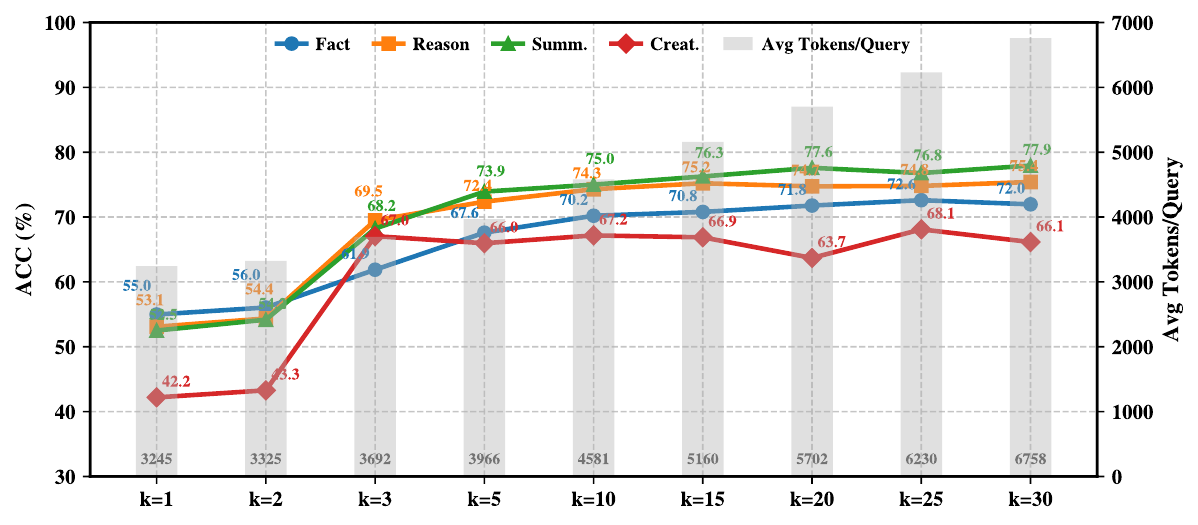}
\caption{Impact of the retrieval Top-$k$ hyperparameter on answer accuracy and token length: Top-$k$ specifies how many knowledge atoms AtomicRAG retrieves per query, and token length is the total number of tokens in the LLM input formed by the question and the retrieved atoms.}
    \label{fig:re}
\end{figure}

\begin{figure}[!htb]
    \centering
    \includegraphics[width=0.43\textwidth]{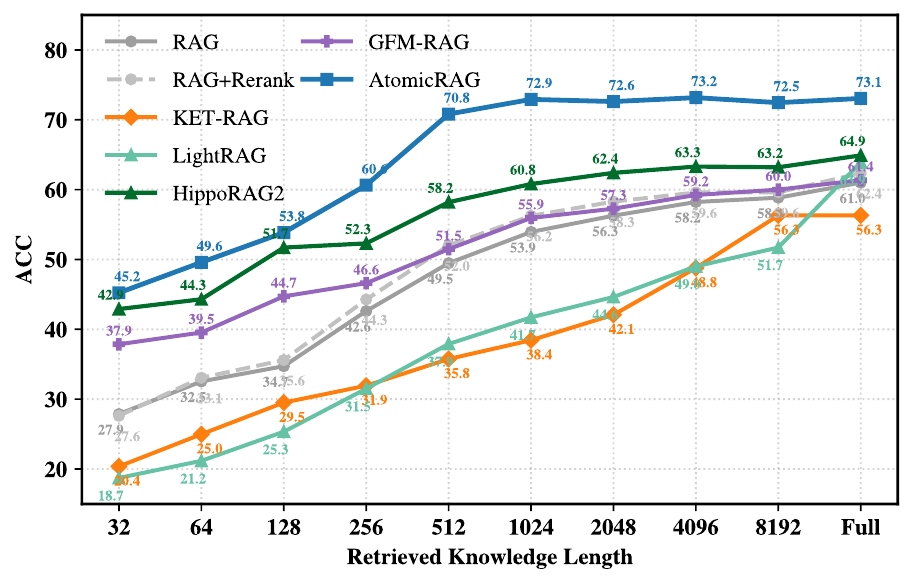}
    \caption{Accuracy under limited context lengths : each point is evaluated with a fixed context budget, defined as the maximum number of tokens permitted in the LLM input, and all methods are truncated to this budget before generation.}
    \label{fig:prompt}
\end{figure}
We evaluate efficiency in terms of (i) the accuracy–token trade-off as Top-$k$ varies, (ii) robustness under fixed context budgets, and (iii) per-query retrieval latency.

\textbf{Effect of Top-$k$.}
In Figure~\ref{fig:re}, the token cost increases monotonically from 3245 at $k{=}1$ to 6230 at $k{=}25$, while accuracy improves rapidly at small $k$ (e.g., \textit{Reason} and \textit{Summ.} rise sharply from $k{=}1$ to $k{=}3$) and then saturates for larger $k$ (roughly $k{\ge}10$), indicating that most decisive evidence is already covered with a moderate Top-$k$ and additional retrieval mainly introduces redundancy.

\textbf{Performance under limited lengths.}
Figure~\ref{fig:prompt} shows that our method remains strong under tight budgets, reaching 70.8 ACC at a context budget of 512 tokens and stabilizing around 73 thereafter, whereas strong baselines improve more gradually and require substantially longer contexts to approach their best performance, demonstrating that our retrieved evidence is denser and more budget-efficient in length-constrained settings.

\textbf{Retrieval latency.}
Table~\ref{tab:efficiency_acc} reports the retrieval latency per query. Our method is the fastest among graph-based baselines, with 0.79 seconds per query, compared to 2.73 for KET-RAG, 4.50 for HippoRAG2, 6.08 for GFM-RAG, and 13.99 for LightRAG. This speedup is consistent with our atom–entity representation: relevance propagation operates on a compact graph, avoiding repeated multi-round expansion in prior systems.

\subsection{Token Overhead (RQ5)}
Table~\ref{tab:efficiency_acc} reports indexing/QA token usage, latency, and ACC on Graph-Bench (Medical). Our method keeps query-time tokens low relative to graph methods with expensive multi-round expansion. While GFM-RAG uses fewer total tokens (3.60M), it is slower (6.08 s/q) and less accurate (61.4 ACC); AtomicRAG achieves the best accuracy with low latency and moderate total tokens. Compared with HippoRAG2, we spend slightly more tokens in indexing (+0.70M) but save more during QA (-1.66M), resulting in a lower total token cost (4.87M vs.\ 5.83M) while improving ACC by 8.3 points. Overall, AtomicRAG incurs a small additional cost during graph construction, but this overhead is compensated by lower QA-time token usage and higher answer accuracy.

\begin{table}[t]
\centering
\footnotesize
\caption{\textbf{Efficiency and performance comparison.} We report token consumption for indexing and QA (in millions, M), overall token cost (indexing + QA), average retrieval latency, and answer accuracy (ACC). Each token entry is shown as \emph{Total} with the \textcolor{gray}{(Prompt + Completion)} breakdown in gray.}
\label{tab:efficiency_acc}
\resizebox{\linewidth}{!}{%
\begin{tabular}{cccccc}
\toprule
\textbf{Method} &
\textbf{Index tokens (M)} &
\textbf{QA tokens (M)} &
\textbf{Total tokens (M)} &
\textbf{Latency (s/q)} &
\textbf{ACC (\%)} \\
\midrule
Vanilla RAG       &
0.00~\textcolor{gray}{(0.00 + 0.00)} &
2.87~\textcolor{gray}{(2.68 + 0.19)} &
2.87 &
0.01 &
62.4 \\
\midrule
KET-RAG   &
2.31~\textcolor{gray}{(1.74 + 0.57)} &
24.45~\textcolor{gray}{(24.42 + 0.03)} &
26.76 &
\underline{2.73} &
47.1 \\
LightRAG  &
2.66~\textcolor{gray}{(2.35 + 0.31)} &
61.51~\textcolor{gray}{(60.74 + 0.77)} &
64.17 &
13.99 &
63.9 \\
HippoRAG2 &
\underline{1.54}~\textcolor{gray}{(1.18 + 0.36)} &
4.29~\textcolor{gray}{(3.96 + 0.34)} &
5.83 &
4.50 &
\underline{64.9} \\
GFM-RAG   &
\textbf{1.52}~\textcolor{gray}{(1.19 + 0.33)} &
\textbf{2.08}~\textcolor{gray}{(1.72 + 0.36)} &
\textbf{3.60} &
6.08 &
61.4 \\
\rowcolor{gray!15}
Ours      &
2.24~\textcolor{gray}{(1.67 + 0.56)} &
\underline{2.63}~\textcolor{gray}{(1.97 + 0.67)} &
\underline{4.87} &
\textbf{0.79} &
\textbf{73.2} \\
\bottomrule
\end{tabular}%
}
\end{table}

\section{Conclusion}
This work identifies a core mismatch between knowledge representation and knowledge indexing in existing RAG systems and introduces AtomicRAG to explicitly decouple the two: semantic content is carried solely by knowledge atoms, while an unlabeled Atom–Entity Graph provides only reachability and aggregation priors rather than encoding predicate semantics. Extensive experiments show that this design yields more stable evidence chains, more compact retrieval contexts, and better accuracy–efficiency trade-offs in multi-hop settings, making AtomicRAG a practical solution for knowledge-intensive, complex queries.
\section*{Impact Statement}
This paper presents work whose goal is to advance the field of 
Machine Learning. There are many potential societal consequences 
of our work, none which we feel must be specifically highlighted here.

\nocite{langley00}

\bibliography{example_paper}
\bibliographystyle{icml2026}

\newpage
\appendix
\onecolumn
\appendix
\onecolumn 


\appendix
\section*{Appendix}
\addcontentsline{toc}{section}{Appendix}
This appendix provides supplementary material that complements the main paper by enabling full reproducibility, isolating the sources of performance gains, and offering both theoretical and qualitative insights. 
Section~\ref{app:repro} documents implementation details and evaluation protocols, including baseline configurations, prompt templates, per-dataset graph statistics, and runtime/token/cost breakdowns, so that all results can be reproduced under consistent settings. 
Section~\ref{app:ablation} reports ablation and sensitivity studies that vary embedding models, LLM backbones, graph retrieval strategies, and key Personalized PageRank (PPR) hyperparameters, clarifying which design choices drive the improvements. 
Section~\ref{appendix:proofs} presents complete proofs of the theoretical claims stated in the main text. 
Section~\ref{app:case} provides qualitative case studies comparing AtomicRAG with representative RAG variants to illustrate typical failure modes and how atomic-level structured retrieval mitigates them. 
Section~\ref{app:prompts} includes the full prompt templates used throughout the pipeline for exact reproducibility. Finally, Section~\ref{app:limitation} analyzes the limitations of AtomicRAG and discusses practical failure modes and directions for improvement.

\section{Reproducibility Details}
\label{app:repro}
\subsection{Baselines and Implementation Details}
\label{app:baselines}
\subsubsection{Baselines}
\label{app:baseline_chunk}

We group baselines into two categories: (i) \textbf{Vanilla RAG}, a standard dense-retrieval pipeline with the same generator, evaluated \emph{without} reranking and \emph{with} reranking; and (ii) \textbf{Graph-enhanced RAG}, representative systems that organize evidence with explicit structures, including MS-GraphRAG (Local/Global), RAPTOR, LightRAG, HippoRAG, HippoRAG2, Fast-GraphRAG, LazyGraphRAG, KET-RAG, KGP, StructRAG, and GFM-RAG.

\paragraph{Vanilla RAG.}
We implement a standard dense retriever over text chunks and use the same generator as AtomicRAG. We report two variants: \textbf{w/o reranking}, which uses the retriever-returned order; and \textbf{w/ reranking}, which reorders retrieved chunks with \texttt{bge-reranker-large} before generation to improve evidence prioritization.

\paragraph{RAPTOR.}
RAPTOR constructs a hierarchical tree index by recursively clustering chunks and summarizing clusters into higher-level nodes. At inference time, it retrieves across multiple abstraction levels (leaf chunks and internal summaries), enabling evidence selection that balances local detail and global context.

\paragraph{MS-GraphRAG (Local/Global).}
MS-GraphRAG organizes the corpus into an entity/community graph and supports two retrieval modes. \textbf{Local} retrieval gathers fine-grained supporting evidence from entity-centric neighborhoods, while \textbf{Global} retrieval aggregates community-level evidence for corpus-level questions, typically via structured aggregation and summarization.

\paragraph{LightRAG.}
LightRAG incorporates lightweight graph organization into indexing and retrieval, supporting both local evidence lookup and higher-level discovery over graph-structured information, with an emphasis on simplicity and efficiency in practical deployments.

\paragraph{HippoRAG.}
HippoRAG builds a schemaless knowledge graph from the corpus and performs associative multi-hop retrieval via graph propagation (e.g., Personalized PageRank). This mechanism improves multi-hop evidence discovery without relying on expensive iterative prompting or heavy inference-time exploration.

\paragraph{HippoRAG2.}
HippoRAG2 extends HippoRAG-style associative retrieval with stronger passage integration and more effective use of LLM modules, targeting improved evidence connectivity and robustness in multi-hop reasoning settings.

\paragraph{Fast-GraphRAG.}
Fast-GraphRAG is an efficiency-oriented GraphRAG variant that accelerates structure-aware retrieval, leveraging fast graph exploration/propagation to identify relevant nodes and passages under practical latency constraints.

\paragraph{LazyGraphRAG.}
LazyGraphRAG reduces up-front indexing and summarization costs by shifting portions of structure-aware retrieval to inference time, using budgeted/on-the-fly exploration to balance computational cost and answer quality.

\paragraph{KET-RAG.}
KET-RAG emphasizes cost-efficient indexing through multi-granular construction: it first selects key chunks to build a lightweight graph skeleton and then leverages an auxiliary structure over the full corpus to support retrieval without fully materializing a dense knowledge graph.

\paragraph{KGP.}
KGP (Knowledge Graph Prompting) uses an LLM-guided traversal process over a passage/structure graph, iteratively navigating graph nodes to gather supporting passages. The graph provides global constraints on evidence transitions for multi-document and multi-hop question answering.

\paragraph{StructRAG.}
StructRAG performs inference-time hybrid structurization: it retrieves raw evidence and then restructures it into a task-appropriate schema or structured context before reasoning, aiming to improve global integration and reduce sensitivity to scattered or noisy evidence.

\paragraph{GFM-RAG.}
GFM-RAG employs a learned graph retriever (e.g., a GNN-based retriever) to reason over graph structure and retrieve relevant evidence, improving robustness compared to purely heuristic traversal or propagation on noisy graphs.

\paragraph{Implementation and evaluation protocol.}
All graph-enhanced baselines are implemented strictly following the configurations and evaluation protocol of \textsc{Graph-Bench} for fair comparison. For the \textbf{Medical} and \textbf{Novel} benchmarks, we directly report the official \textsc{Graph-Bench} results, as these datasets are evaluated under fixed standardized settings. For \textbf{multi-hop} tasks, we faithfully reimplement each baseline using the corresponding \textsc{Graph-Bench} configurations for indexing, retrieval, and inference, and evaluate them under the same experimental conditions as AtomicRAG.

\subsubsection{AtomicRAG Default Configuration}
\label{app:default_config}

\paragraph{Backbone models.}
We use \texttt{gpt-4o-mini} as the generator LLM and \texttt{BAAI/bge-large-en-v1.5} as the embedding model. The embedding maximum sequence length is set to 2048.

\begin{table}[t]
\centering
\small
\setlength{\tabcolsep}{4.5pt}
\begin{tabular}{l c p{6.6cm}}
\toprule
\textbf{Component / Parameter} & \textbf{Value} & \textbf{Description} \\
\midrule
LLM & \texttt{gpt-4o-mini} & Generator used for final answer synthesis. \\
Embedding model & \texttt{BAAI/bge-large-en-v1.5} & Dense encoder for atoms/passages and queries. \\
embedding\_max\_seq\_len & 2048 & Maximum input length for the embedding model. \\
retrieval\_top\_k & 25 & Number of retrieved candidates atoms per query for downstream selection. \\
\midrule
synonymy\_edge\_topk & 2047 & Top-$k$ nearest neighbors used to construct synonymy edges (KNN). \\
synonymy\_edge\_sim\_threshold & 0.8 & Minimum similarity required to add a synonymy edge. \\
\midrule
entity\_node\_weight & 1.0 & Weight factor for entity seeds when initializing propagation. \\
entity\_top\_k & 20 & Max number of entity nodes retained per query as initial seeds. \\
entity\_sim\_threshold & 0.3 & Minimum similarity for an entity to be considered a valid seed. \\
\midrule
propagation\_method & \texttt{ppr} & Graph propagation method (Personalized PageRank). \\
damping & 0.3 & PPR damping factor controlling restart probability. \\
passage\_node\_weight & 0.1 & Weight assigned to passage/atom nodes in the propagation graph. \\
propagation\_num\_iter & 20 & Iterations for iterative propagation. \\
propagation\_num\_walks & 1000 & Number of random walks used for Monte-Carlo PPR estimation. \\
propagation\_walk\_length & 10 & Length of each random walk. \\
\midrule
max\_sub\_questions & 3 & Maximum number of induced sub-questions per query . \\
complexity\_threshold & 6.5 & Threshold for triggering query decomposition. \\
\bottomrule
\end{tabular}
\caption{Hyperparameters and default settings used in our AtomicRAG implementation.}
\label{tab:hyperparams}
\end{table}
 We set \texttt{retrieval\_top\_k=25} for candidate retrieval. For multi-hop queries, we use a conservative decomposition budget (\texttt{max\_sub\_questions=3}) and raise the decomposition trigger threshold (\texttt{complexity\_threshold=6.5}) to avoid over-fragmenting simple queries while still enabling decomposition on genuinely complex questions. For graph retrieval, we use Personalized PageRank (\texttt{propagation\_method=ppr}) with \texttt{damping=0.3}; synonymy edges are constructed via KNN with \texttt{synonymy\_edge\_topk=2047} and filtered by \texttt{synonymy\_edge\_sim\_threshold=0.8}.
\subsubsection{Per-Dataset Graph Statistics}
\label{app:graph-stats}

Table~\ref{tab:aeg_stats} summarizes the graph-level statistics of the constructed Atom--Entity Graph (AEG) for each dataset. In our AEG, \textbf{nodes} are the union of \textbf{entity nodes} and \textbf{atomic knowledge nodes} (knowledge atoms). \textbf{Edges} are partitioned into three types: (i) \textbf{related} edges linking entity--entity pairs (capturing co-occurrence/association), (ii) \textbf{synonym} edges connecting semantically similar entities (constructed via embedding KNN and thresholding), and (iii) \textbf{containment} edges linking atoms to the entities they mention (atom--entity incidence). These statistics provide a concrete view of corpus-dependent graph size and sparsity, which directly affect indexing cost and propagation-based retrieval efficiency.

Across datasets, the AEG size scales with corpus complexity: multi-hop QA datasets such as HotpotQA and MuSiQue yield the largest graphs (both in nodes and total edges), reflecting broader entity coverage and denser inter-entity connectivity. In contrast, Medical exhibits substantially fewer nodes but a comparatively high edge count per node, indicating a more densely connected entity space under synonymy and containment relations. Novel shows moderate graph size with balanced edge composition, consistent with narrative-style corpora that introduce many entities but comparatively fewer cross-document associations than encyclopedic QA benchmarks.

\begin{table}[t]
\centering
\small
\setlength{\tabcolsep}{5pt}
\begin{tabular}{lrrrrrrr}
\toprule
\textbf{Dataset} & \textbf{\#Nodes} & \textbf{\#Entities} & \textbf{\#Atoms} & \textbf{\#Edges} & \textbf{\#Related} & \textbf{\#Synonym} & \textbf{\#Containment} \\
\midrule
HotpotQA           & 112,995 & 84,592 & 28,403 & 479,895 & 177,734 & 276,577 & 25,584 \\
MuSiQue            & 118,218 & 87,446 & 30,772 & 488,497 & 175,794 & 286,923 & 25,780 \\
2WikiMultiHop      & 59,782  & 44,750 & 15,032 & 234,008 & 93,446  & 125,422 & 15,140 \\
Medical            & 14,647  & 9,087  & 5,560  & 143,919 & 21,068  & 99,707  & 23,144 \\
Novel  & 59,807  & 40,217 & 19,590 & 176,948 & 75,665  & 83,481  & 17,802 \\
\bottomrule
\end{tabular}
\caption{Per-dataset statistics of the Atom--Entity Graph (AEG). Nodes consist of entity nodes and atomic knowledge nodes. Edges include entity--entity related edges, synonym edges, and atom--entity containment edges.}
\label{tab:aeg_stats}
\end{table}
\subsubsection{Runtime, Token, and Cost Breakdown}
\label{app:cost}

Table~\ref{tab:atomicrag_runtime_token_cost} reports a detailed breakdown of runtime, token usage, and monetary cost across the major modules of \textsc{AtomicRAG}. The results highlight a clear separation between \emph{graph-centric} and \emph{LLM-centric} costs.

From a runtime perspective, \textbf{Entity-Resonance Graph Retrieval} dominates end-to-end latency, accounting for over one-third of total execution time. This reflects the cost of large-scale graph traversal and propagation, which is compute-intensive but largely independent of LLM usage. In contrast, modules involving heavy LLM interaction—such as \textbf{Atomic Sieve} and \textbf{Grounded Answer Generation}—consume a smaller share of wall-clock time despite extensive prompting.

From a token and cost perspective, the \textbf{Atomic Sieve} is the primary contributor, responsible for more than 60\% of total token consumption and cost. This is expected, as the sieve performs fine-grained, fragment-level relevance filtering with long prompt contexts. Atom--Entity Graph Construction and Atomic Question Decomposition incur moderate, one-time or query-level LLM costs, while Entity-Resonance Graph Retrieval introduces no LLM token overhead.

Overall, this breakdown demonstrates that AtomicRAG’s computational cost is dominated by a small number of interpretable stages: graph propagation for latency and LLM-based atom filtering for monetary cost. This modular separation enables targeted optimization—e.g., accelerating graph traversal or pruning candidate atoms before sieving—without redesigning the entire pipeline.

\begin{table}[t]
\centering
\scriptsize
\setlength{\tabcolsep}{2.6pt}
\renewcommand{\arraystretch}{1.10}
\begin{tabular}{>{\raggedright\arraybackslash}m{3.8cm} rr | rrrr | rrr}
\toprule
\textbf{Module (AtomicRAG)}
& \multicolumn{2}{c}{\textbf{Runtime}}
& \multicolumn{4}{c}{\textbf{Tokens}}
& \multicolumn{3}{c}{\textbf{Cost (USD)}} \\
\cmidrule(lr){2-3}\cmidrule(lr){4-7}\cmidrule(lr){8-10}
& \textbf{Time (s)} & \textbf{Share}
& \textbf{Prompt} & \textbf{Completion} & \textbf{Total} & \textbf{Share}
& \textbf{Prompt} & \textbf{Completion} & \textbf{Total} \\
\midrule
Atom--Entity Graph Construction
& 188.84 & 9.4\%
& 1,673,875 & 563,101 & 2,236,976 & 13.8\%
& 0.251 & 0.338 & 0.589 \\

Atomic Question Decomposition
& 111.60 & 5.5\%
& 739,474 & 303,724 & 1,043,198 & 6.5\%
& 0.111 & 0.182 & 0.293 \\

Entity-Resonance Graph Retrieval
& 735.80 & 36.5\%
& -- & -- & -- & --
& -- & -- & -- \\

Atomic Sieve
& 210.69 & 10.4\%
& 10,098,938 & 158,475 & 10,257,413 & 63.4\%
& 1.515 & 0.095 & 1.610 \\

Grounded Answer Generation
& 210.33 & 10.4\%
& 1,967,021 & 666,828 & 2,633,849 & 16.3\%
& 0.295 & 0.400 & 0.695 \\
\midrule
\textbf{Total}
& \textbf{2,018.46} & \textbf{100\%}
& \textbf{14,479,308} & \textbf{1,692,128} & \textbf{16,171,436} & \textbf{100\%}
& \textbf{2.172} & \textbf{1.015} & \textbf{3.187} \\
\bottomrule
\end{tabular}
\caption{Module-level runtime, token consumption, and monetary cost of \textsc{AtomicRAG}. ``Prompt'' and ``Completion'' denote the input and output tokens of the underlying LLM calls, respectively. Costs are reported in USD.}
\label{tab:atomicrag_runtime_token_cost}
\end{table}
\subsubsection{Evaluation Metrics}
\label{app:metrics}

We follow the evaluation protocol in Graph-Bench and report the standard generation-quality metrics. In particular, we adopt \textbf{Answer Accuracy} as the primary accuracy metric, which jointly evaluates semantic alignment and factual correctness to avoid over-rewarding answers that are fluent but hallucinated, or factual but semantically mismatched. Notably, this metric combines an \emph{LLM-based claim-level verifier} with \emph{embedding-based semantic similarity}, yielding a hybrid scoring function that reflects both factual faithfulness and semantic alignment (the verifier prompt templates are provided in Graph-Bench).

\paragraph{Answer Accuracy.}
Answer Accuracy (ACC) provides a dual assessment of answer quality by combining (i) \emph{semantic similarity} and (ii) \emph{fine-grained factual verification}:
\begin{equation}
\mathrm{ACC} = \alpha \cdot \mathrm{FC} + (1-\alpha)\cdot \mathrm{SS},
\label{eq:ac}
\end{equation}
where $\alpha$ is a weighting coefficient (we use $\alpha=0.7$ by default following Graph-Bench).

\paragraph{Factual Correctness.}
Factual correctness $\mathrm{FC}$ is computed via statement-level verification and summarized as an F1-style score:
\begin{equation}
\mathrm{FC}=\frac{2\cdot \mathrm{TP}}{2\cdot \mathrm{TP}+\mathrm{FP}+\mathrm{FN}},
\label{eq:fc}
\end{equation}
where $\mathrm{TP}$ is the number of verified correct claims, $\mathrm{FP}$ is the number of incorrect (hallucinated) claims, and $\mathrm{FN}$ is the number of missing reference claims not covered by the generated answer. This formulation explicitly penalizes both hallucinations (FP) and omissions (FN). In Graph-Bench, $\mathrm{TP}/\mathrm{FP}/\mathrm{FN}$ are obtained by prompting an LLM to verify fine-grained claims against the reference evidence; see the official Graph-Bench prompt templates for the exact verifier instructions.

\paragraph{Semantic Similarity.}
Semantic similarity $\mathrm{SS}$ is measured by embedding-based cosine similarity:
\begin{equation}
\mathrm{SS} = \cos\!\left(\mathbf{e}(\hat{y}),\, \mathbf{e}(y)\right),
\label{eq:ss}
\end{equation}
where $\hat{y}$ and $y$ denote the generated answer and the reference answer, respectively, and $\mathbf{e}(\cdot)$ maps text to its embedding representation. This term rewards semantic alignment even when surface forms differ.

\section{Ablation and Sensitivity Analyses}
\label{app:ablation}
\subsection{Embedding Model Ablations}
\label{app:ablation_embedding}
We ablate the embedding model while keeping the Atom--Entity Graph construction, retrieval, and generation prompts fixed. Table~\ref{tab:embedding_models} shows that performance is moderately sensitive to the embedding choice: stronger general-purpose English embeddings consistently improve both fact retrieval and multi-hop reasoning, which in turn lifts the overall average. Among all candidates, \texttt{bge-large-en-v1.5} achieves the best Avg. score and leads on three out of four categories, indicating that our retrieval pipeline benefits most from embeddings with high semantic separability at the atomic-text granularity. In contrast, models with competitive summarization scores (e.g., \texttt{nomic-embed}) do not necessarily translate to stronger factuality or reasoning, suggesting that optimizing for long-form semantic similarity alone is insufficient for evidence selection in multi-hop settings. Unless stated otherwise, we adopt \texttt{bge-large-en-v1.5} in all experiments.

\begin{table}[t]\centering\small\setlength{\tabcolsep}{6pt}
\begin{tabular}{lccccc}\toprule
Embedding Model & Fact & Reason & Summ. & Creat. & Avg. \\\midrule
\rowcolor{gray!15} bge-large-en-v1.5 & \textbf{72.6} & \textbf{74.8} & \underline{76.8} & \textbf{68.3} & \textbf{73.1} \\
gte-qwen2 (1.5B) & 70.9 & \underline{74.3} & 76.2 & \underline{66.7} & \underline{72.0} \\
gte-large & \underline{71.4} & 73.7 & 76.2 & 66.3 & 71.9 \\
nomic-embed & 71.1 & 73.4 & \textbf{77.3} & 65.7 & 71.9 \\
bge-m3 & 71.7 & \textbf{74.3} & 75.1 & 66.4 & 71.8 \\
e5-large & 70.2 & 71.5 & 75.1 & 66.7 & 70.9 \\\bottomrule
\end{tabular}
\caption{\textbf{Ablation on embedding models.}  Best results are in bold and second-best results are underlined.}
\label{tab:embedding_models}
\end{table}

\subsection{LLM Backbone Ablations}\label{app:ablation_llm}
We study the effect of the generator by keeping the retrieval pipeline unchanged and replacing only the LLM backbone. Table~\ref{tab:llm_ablation} reports per-category results and the overall average. Overall, AtomicRAG remains robust across backbones: stronger instruction-tuned models yield consistent improvements, but the relative ranking across categories is stable, implying that the primary gains come from retrieval quality rather than model-specific prompting quirks. Notably, \texttt{GPT-4o-mini} attains the highest Avg., while other backbones incur predictable degradations that correlate with model capacity; however, even with smaller backbones (e.g., \texttt{Qwen2.5-7B} and \texttt{Llama-3.1-8B}), the method maintains reasonable performance, indicating that the retrieved atomic evidence is sufficiently targeted to reduce the burden on the generator.

\begin{table}[t]
\centering
\small
\setlength{\tabcolsep}{6pt}
\renewcommand{\arraystretch}{1.15}
\begin{tabular}{lccccc}
\toprule
\textbf{Method} & \textbf{Fact} & \textbf{Reason} & \textbf{Summ.} & \textbf{Creat.} & \textbf{Avg.} \\
\midrule
RAG (w/ rerank)      & 57.6 & 56.0 & 62.0 & 60.9 & 59.1 \\
\midrule
GraphRAG (local)     & 49.2 & 61.6 & 59.0 & 61.7 & 57.9 \\
GraphRAG (global)    & 40.0 & 61.4 & 51.1 & 59.2 & 52.9 \\
HippoRAG2            & 64.5 & 64.1 & 64.7 & 60.8 & 63.5 \\
LightRAG             & 64.4 & 64.7 & 69.4 & 63.3 & 65.4 \\
Fast-GraphRAG        & 62.0 & 62.4 & 65.1 & 63.0 & 63.1 \\
RAPTOR               & 50.5 & 52.9 & 58.9 & 61.5 & 55.9 \\
\midrule
\rowcolor{gray!15} \textbf{Ours} & \textbf{68.2} & \textbf{71.1} & \textbf{74.9} & \textbf{68.9} & \textbf{70.8} \\
\rowcolor{gray!15} \textit{Improv.\ vs best baseline} & $\uparrow$3.7 & $\uparrow$6.4 & $\uparrow$5.5 & $\uparrow$5.7 & $\uparrow$5.3 \\
\bottomrule
\end{tabular}
\caption{ACC on the Medical dataset using Qwen2.5-14B for generation-based evaluation. Avg. is the mean over four categories. \textit{Improv.\ vs best baseline} reports absolute score gains (in points) of Ours over the best baseline (excluding Ours); $\uparrow$ denotes increases.}
\label{tab:medical_qwen_baselines}
\end{table}


\newcommand{\vendoricon}[2]{\raisebox{-0.18em}{\includegraphics[height=1.15em]{#1}}~#2}

\begin{table}[t]
\centering
\footnotesize
\setlength{\tabcolsep}{4.0pt}
\renewcommand{\arraystretch}{1.15}
\begin{tabular}{l l ccccc c}
\toprule
\textbf{Vendor} & \textbf{LLM} & \textbf{Fact} & \textbf{Reason} & \textbf{Summ.} & \textbf{Creat.} & \textbf{Avg.} & $\Delta$Avg \\
\midrule
\rowcolor{gray!12}
\vendoricon{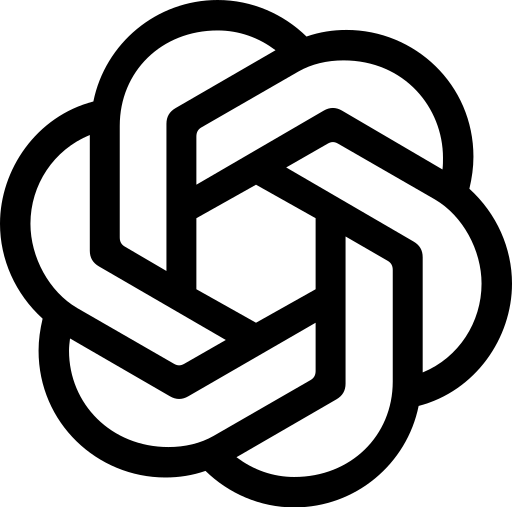}{OpenAI}   & GPT-4o-mini            & \textbf{72.6} & \textbf{74.8} & \underline{76.8} & \underline{68.3} & \textbf{73.1} & -- \\
\vendoricon{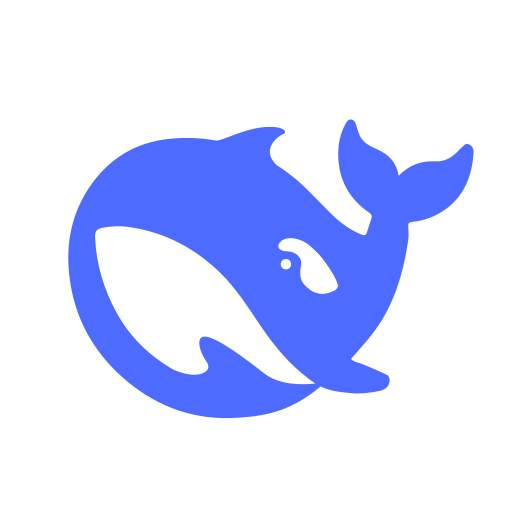}{DeepSeek} & DeepSeek-V3          & \underline{69.3} & \underline{73.0} & \textbf{79.5} & 67.0 & \underline{72.2} & -0.9 \\
\vendoricon{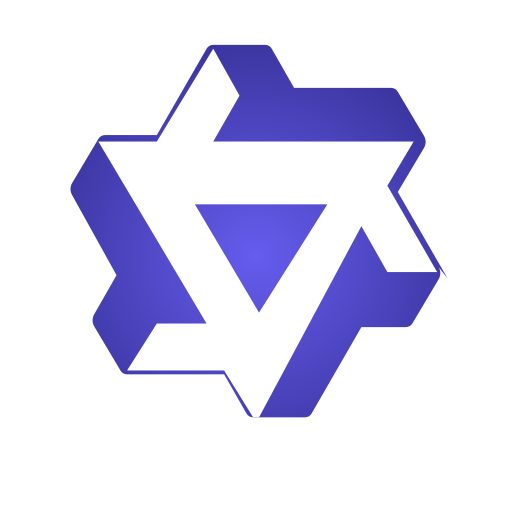}{Qwen}       & Qwen2.5-14B-Instruct   & 68.2 & 71.1 & 74.9 & \textbf{68.9} & 70.8 & -2.3 \\
\vendoricon{fig/logo/qwen.png}{Qwen}       & Qwen2.5-7B-Instruct    & 64.7 & 69.2 & 73.2 & 65.3 & 68.1 & -5.0 \\
\raisebox{-0.2em}{\includegraphics[height=0.95em]{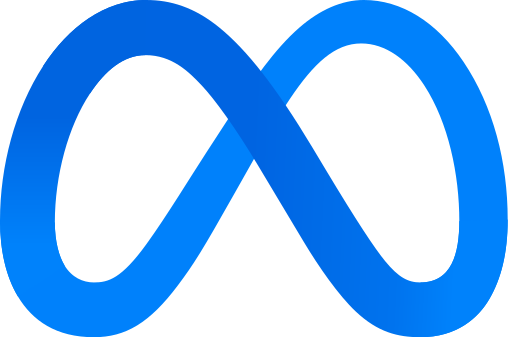}}~Meta
& Llama-3.1-8B-Instruct  & 64.7 & 67.8 & 70.8 & 62.5 & 66.5 & -6.6 \\
\bottomrule
\end{tabular}
\caption{\textbf{Ablation on LLM backbones.} We report per-category scores and the overall average (Avg.). $\Delta$Avg denotes the absolute point change relative to GPT-4o-mini.}
\label{tab:llm_ablation}
\end{table}

\begin{table}[t]
\centering
\footnotesize
\setlength{\tabcolsep}{4.2pt}
\renewcommand{\arraystretch}{1.15}
\begin{tabular}{lccccccc}
\toprule
\textbf{Retriever} & \textbf{Fact} & \textbf{Reason} & \textbf{Summ.} & \textbf{Creat.} & \textbf{Avg.} & $\Delta$Avg & \textbf{Time (s/q)} \\
\midrule
\rowcolor{gray!12}
PPR (Baseline)     & \textbf{72.6} & \textbf{74.8} & 76.8 & \textbf{68.3} & \textbf{73.1} & --   & \underline{0.79} \\
RWR                & 62.4 & 70.9 & 74.5 & 65.0 & 68.2 & -4.9 & \textbf{0.78} \\
Power Iteration    & 70.7 & \underline{73.3} & \textbf{78.6} & \underline{67.9} & \underline{72.6} & -0.5 & 18.67 \\
Katz Index         & \underline{70.7} & 73.0 & \underline{76.9} & 65.7 & 71.6 & -1.5 & 20.28 \\
Label Propagation  & 67.2 & 71.5 & 71.9 & 65.5 & 69.0 & -4.1 & 1.38 \\
Weighted BFS       & 61.0 & 68.1 & 69.8 & 64.3 & 65.8 & -7.3 & 13.37 \\
\bottomrule
\end{tabular}
\caption{\textbf{Ablation on graph retrieval strategies.} We report per-category scores and the overall average (Avg.). lower is better for Time. $\Delta$Avg denotes the absolute point change relative to PPR.}
\label{tab:retrieval_ablation}
\end{table}
Table~\ref{tab:llm_ablation} shows that \textbf{GPT-4o-mini} achieves the best overall score (Avg.\ 73.1), suggesting that generator strength remains a key factor even when retrieval is held constant.
DeepSeek-V3 is the closest alternative (72.2, only 0.9 points behind) and attains the best summarization performance (79.5), indicating stronger long-form synthesis once relevant evidence is provided, albeit with slightly weaker fact- and reasoning-centric results than GPT-4o-mini.
Qwen2.5-14B-Instruct ranks third (70.8) but delivers the highest creative generation score (68.9), implying that different model priors may favor open-ended, stylistic completion over strict grounding.
In contrast, the two smaller instruction models exhibit a clear drop (68.1 and 66.5), with the most pronounced degradation on Complex Reasoning and Creative Generation, consistent with reduced robustness in multi-step evidence integration and global coherence under identical prompting and retrieved context.
\subsection{Additional Baseline Comparison on Medical}
\label{app:medical_baselines_qwen}

To complement the main results that use GPT-4o-mini as the generator and evaluator, we further compare AtomicRAG with representative RAG and graph-enhanced baselines on the \textbf{Medical} split under a unified evaluator. Table~\ref{tab:medical_qwen_baselines} reports Answer Accuracy (ACC) when all methods are evaluated using Qwen2.5-14B for generation-based judging under the official Graph-Bench protocol. This setting fixes the evaluation backbone across systems, so that performance differences more directly reflect the quality of their retrieval and evidence-organization pipelines rather than idiosyncrasies of different LLM judges.

Overall, AtomicRAG achieves the highest Avg.\ score (70.8), outperforming the best baseline, LightRAG (65.4), by more than five absolute points. The gains are consistent across all four categories: AtomicRAG improves over the strongest competing method by several points on Fact Retrieval, Complex Reasoning, Contextual Summarization, and Creative Generation, as summarized in the \emph{Improv.\ vs best baseline} row. These results indicate that, even when judged by a strong external LLM under standardized evaluation, atomic-level evidence organization and entity-resonance retrieval provide more reliable evidence chains than chunk-based or predicate-centric graph baselines on domain-specific medical QA.

Table~\ref{tab:medical_qwen_baselines} lists RAG with reranking, two GraphRAG retrieval modes (local/global), HippoRAG2, LightRAG, Fast-GraphRAG, RAPTOR, and our AtomicRAG. The \emph{Improv.\ vs best baseline} row reports absolute ACC gains (in points) of AtomicRAG over the strongest baseline in each column, offering a concise view of the margin under this unified evaluation protocol.

\subsection{Graph Retrieval Variants}
\label{app:ablation_graph_retrieval}

We compare representative graph diffusion and path-based ranking strategies for retrieving evidence atoms on the entity--atom graph.
Random Walk with Restart (RWR) approximates PPR via Monte Carlo restartable walks from query seeds.
Power Iteration explicitly solves the PPR fixed-point equation through iterative updates until convergence.
Katz Index ranks nodes by counting seed-to-node paths with exponential decay by path length.
Label Propagation performs iterative label diffusion (smoothing) from seeded nodes across the graph.
Weighted BFS conducts hop-based expansion with distance-decay weights, yielding a heuristic diffusion score.

Table~\ref{tab:retrieval_ablation} shows that \textbf{PPR} achieves the best quality--latency trade-off, obtaining the highest overall score (Avg.\ 73.1) with low retrieval latency (0.79 s/query).
RWR matches PPR in latency (0.78 s/query) but incurs a large accuracy drop (Avg.\ 68.2, $\Delta$Avg $=-4.9$), indicating that under our sampling budget the Monte Carlo estimator yields higher-variance and less stable rankings, which is detrimental for evidence-centric tasks.
Power Iteration and Katz remain competitive in score (Avg.\ 72.6 and 71.6), and Power Iteration even improves summarization (78.6 vs.\ 76.8), but both are prohibitively slow online (18.67 and 20.28 s/query), suggesting that per-query convergence and long-range aggregation dominate runtime.
Label Propagation is relatively efficient (1.38 s/query) yet underperforms on average (Avg.\ 69.0, $\Delta$Avg $=-4.1$), consistent with over-smoothing effects that dilute sharp relevance signals.
Weighted BFS performs the worst overall (Avg.\ 65.8) while still being slow (13.37 s/query), implying that hop-based heuristics neither approximate the stationary distribution reliably nor scale well when the search frontier expands.
Overall, this ablation supports using \textbf{stationary-distribution ranking (PPR)} as the default retriever: sampling approximations (RWR) sacrifice too much accuracy, exact solvers/path counting (Power Iteration, Katz) are computationally mismatched to online retrieval, and alternative diffusion schemes (Label Propagation, Weighted BFS) tend to blur discriminative evidence.

\begin{table}[t]
\centering
\small
\setlength{\tabcolsep}{5pt}

\begin{minipage}[t]{0.49\linewidth}
\centering
\setlength{\tabcolsep}{6pt}
\begin{tabular}{cccccc}
\toprule
Damping ($\rho$) & Fact & Reason & Summ. & Creat. & Avg. \\
\midrule
0.1 & \underline{71.8} & 72.6 & \textbf{77.2} & 67.2 & \underline{72.2} \\
0.2 & 69.8 & 72.5 & 75.7 & 66.3 & 71.1 \\
\rowcolor{gray!15}
0.3 & \textbf{72.6} & \textbf{74.8} & \underline{76.8} & \textbf{68.3} & \textbf{73.1} \\
0.4 & 70.5 & \underline{72.9} & 74.9 & 66.1 & 71.1 \\
0.5 & 71.3 & \underline{73.2} & 76.5 & 65.6 & 71.7 \\
0.6 & 69.8 & 72.3 & 75.1 & 67.2 & 71.1 \\
0.7 & 71.9 & 72.9 & 76.5 & \underline{67.9} & 72.3 \\
0.8 & 65.1 & 70.6 & 69.6 & 64.8 & 67.5 \\
0.9 & 61.7 & 67.6 & 69.2 & 66.4 & 66.2 \\
\bottomrule
\end{tabular}
\caption{Ablation on restart/damping coefficient $\rho$. Best in bold and second-best underlined.}
\label{tab:ablation_damping}
\end{minipage}\hfill
\begin{minipage}[t]{0.49\linewidth}
\centering
\setlength{\tabcolsep}{6pt}
\begin{tabular}{cccccc}
\toprule
Atom weight ($\lambda_{\text{seed}}$) & Fact & Reason & Summ. & Creat. & Avg. \\
\midrule
0.01 & \underline{71.1} & 72.8 & \underline{76.2} & \underline{67.5} & \underline{71.9} \\
0.05 & 71.3 & \underline{73.2} & \textbf{76.5} & 65.6 & 71.7 \\
\rowcolor{gray!15}
0.1 & \textbf{72.6} & \textbf{74.8} & \underline{76.8} & \textbf{68.3} & \textbf{73.1} \\
0.2 & 70.0 & 72.7 & 75.1 & 64.4 & 70.6 \\
0.3 & 69.6 & 72.9 & 74.7 & 65.2 & 70.6 \\
0.4 & 69.9 & 72.6 & 74.4 & 64.4 & 70.3 \\
0.5 & 69.2 & 72.1 & 73.3 & 64.0 & 69.7 \\
0.6 & 68.4 & 70.9 & 73.8 & 60.4 & 68.4 \\
0.7 & 64.2 & 66.1 & 69.5 & 60.9 & 65.2 \\
0.8 & 62.7 & 67.7 & 69.7 & 59.4 & 64.9 \\
0.9 & 59.5 & 62.0 & 64.1 & 58.7 & 61.1 \\
1.0 & 59.0 & 62.2 & 62.6 & 55.3 & 59.8 \\
\bottomrule
\end{tabular}
\caption{Ablation on atom weight $\lambda_{\text{seed}}$. Best in bold and second-best underlined.}
\label{tab:ablation_lambda}
\end{minipage}

\end{table}

\subsection{PPR Hyperparameter Sensitivity}
\label{app:ppr_hparams}
We analyze the sensitivity of our PPR-based resonance retrieval to two key hyperparameters: the restart/damping coefficient $\rho$ and the atom-seed personalization weight $\lambda_{\text{seed}}$. In all runs, we keep the graph, query decomposition, reranking, and generation settings fixed, and vary only the target hyperparameter.

\paragraph{Damping coefficient $\rho$.}
Table~\ref{tab:ablation_damping} shows a clear ``sweet spot'' around $\rho=0.3$, which yields the best overall Avg. score. Smaller values (e.g., $\rho\in[0.1,0.2]$) emphasize local neighborhoods around the seed distribution, which can preserve salient evidence for summarization but tends to under-explore longer multi-hop paths, limiting gains on reasoning. As $\rho$ increases beyond $0.5$, the walk becomes increasingly dominated by graph diffusion; while this can slightly help coverage in some cases, it also amplifies connectivity noise and dilutes query-specific focus. This effect becomes pronounced for large $\rho$ (e.g., $\rho\ge 0.8$), where performance drops substantially across all categories, consistent with over-smoothing toward high-degree or globally central entities.

\paragraph{Atom-seed weight $\lambda_{\text{seed}}$.}
Table~\ref{tab:ablation_lambda} varies the mass assigned to atom-derived seeds in the personalization vector. We find that moderate seeding ($\lambda_{\text{seed}}=0.1$) is optimal and robust, indicating that atom-level signals should strongly guide propagation but still leave room for entity-level diffusion to bridge hops. When $\lambda_{\text{seed}}$ is too small (e.g., $0.01$), PPR relies more heavily on coarse entity connectivity, reducing selectivity and hurting fact/creativity. Conversely, overly large $\lambda_{\text{seed}}$ (e.g., $\ge 0.3$) makes the walk overly myopic: rankings become dominated by a narrow set of seed-adjacent atoms, which degrades cross-entity aggregation and harms multi-hop retrieval, leading to a monotonic decline as $\lambda_{\text{seed}}\rightarrow 1.0$.

Overall, these results suggest that AtomicRAG benefits from a balanced regime where PPR propagation is neither too local nor too global, and where atom-level personalization provides a strong but not exclusive anchor for multi-hop evidence discovery. Unless stated otherwise, we use $\rho=0.3$ and $\lambda_{\text{seed}}=0.1$.


\section{Proofs}
\label{appendix:proofs}

\subsection{Proof of Proposition 1: AEG is more comprehensive and robust than predicate-labeled Knowledge Graph}
\label{proof1}

This section formalizes the representational and robustness advantages of the Atom--Entity Graph (AEG) used in AtomicRAG.
We compare against \emph{predicate-labeled knowledge graphs} commonly used in graph-based RAG, where semantic content is carried by
extracted predicate-typed edges.

\paragraph{Predicate-labeled KG baseline.}
A predicate-labeled knowledge graph is a directed, predicate-typed multigraph
\[
G_{\mathrm{KG}}=(\mathcal{E},\mathcal{R},\mathcal{T}), \qquad \mathcal{T}\subseteq \mathcal{E}\times \mathcal{R}\times \mathcal{E},
\]
where each triple $(h,r,t)\in\mathcal{T}$ is interpreted as a relational assertion $r(h,t)$.
This baseline does \emph{not} include qualifiers, reification, provenance nodes, or higher-order statements.

\paragraph{AEG recap.}
AEG is a heterogeneous graph $G_{\mathrm{AEG}}=(V,\mathcal{L})$ with node set $V=\mathcal{A}\cup\mathcal{E}$.
Semantic content is carried exclusively by atoms $a\in\mathcal{A}$, where each atom is a minimal, self-contained proposition.
Graph edges provide only organization: the backbone containment edges
\[
\mathcal{L}_{\mathrm{cont}}=\{(a,e)\mid a\in\mathcal{A},\, e\in \mathcal{E}(a)\}
\]
encode only the structural fact that entity $e$ is mentioned in atom $a$.
Optional auxiliary entity--entity links serve as weak connectivity cues and are not predicate-typed commitments.

\subsubsection{Comprehensiveness: embedding KG into AEG and strictness}

\begin{definition}[KG-to-AEG embedding]
\label{def:kg_to_aeg}
Given a predicate-labeled KG $G_{\mathrm{KG}}=(\mathcal{E},\mathcal{R},\mathcal{T})$, define an AEG
$\Phi(G_{\mathrm{KG}})=(V,\mathcal{L})$ as follows.
For every triple $(h,r,t)\in\mathcal{T}$, create an atom node $a_{h,r,t}\in\mathcal{A}$ whose atom text encodes the proposition $r(h,t)$,
and set $\mathcal{E}(a_{h,r,t})=\{h,t\}$.
Add containment edges $(a_{h,r,t},h)$ and $(a_{h,r,t},t)$ to $\mathcal{L}_{\mathrm{cont}}$.
No predicate-typed edge is added to $\mathcal{L}$.
\end{definition}

\begin{definition}[AEG-to-KG projection]
\label{def:aeg_to_kg}
Define a projection $\pi_{\mathrm{KG}}$ that maps $\Phi(G_{\mathrm{KG}})$ back to a predicate-labeled KG
by parsing each atom $a_{h,r,t}$ into the corresponding triple $(h,r,t)$ and returning the set of all such triples.
\end{definition}

\begin{lemma}[AEG can represent any predicate-labeled KG]
\label{lem:kg_embeddable}
For any predicate-labeled KG $G_{\mathrm{KG}}$, we have $\pi_{\mathrm{KG}}(\Phi(G_{\mathrm{KG}}))=G_{\mathrm{KG}}$.
\end{lemma}

\begin{proof}
By construction, for each $(h,r,t)\in\mathcal{T}$, $\Phi$ creates exactly one atom $a_{h,r,t}$ encoding proposition $r(h,t)$.
The projection $\pi_{\mathrm{KG}}$ recovers all such triples and no others, hence
$\pi_{\mathrm{KG}}(\Phi(G_{\mathrm{KG}}))=(\mathcal{E},\mathcal{R},\mathcal{T})=G_{\mathrm{KG}}$.
\end{proof}

Lemma~\ref{lem:kg_embeddable} shows that AEG is at least as expressive as the predicate-labeled KG baseline for representing relational assertions.

We now establish strictness by exhibiting information that AEG represents natively (as distinct atomic propositions that preserve local semantic context)
but the predicate-labeled KG baseline cannot represent without extending the formalism beyond predicate-typed edges.

\begin{definition}[Contextual distinguishability]
\label{def:contextual}
A representation is \emph{contextually distinguishable} if it can represent two evidences that share the same relational core (same $(h,r,t)$)
but differ in contextual semantic content (e.g., time, scope, attribution, or discourse-resolved qualifiers) as distinct objects, without collapsing them.
\end{definition}

\begin{lemma}[Predicate-labeled KG is not contextually distinguishable]
\label{lem:kg_not_contextual}
In the predicate-labeled KG baseline, two evidences with the same relational core $(h,r,t)$ are necessarily identified as the same edge assertion,
and thus their contextual distinction is lost unless the model is extended beyond predicate-typed edges.
\end{lemma}

\begin{proof}
In the baseline KG, the semantic carrier is the typed edge $(h,r,t)\in\mathcal{E}\times\mathcal{R}\times\mathcal{E}$.
If two evidences share the same $(h,r,t)$, they map to the same element of $\mathcal{T}$.
The baseline structure has no additional components to encode distinct contexts while keeping them distinct, unless one introduces extra objects
(e.g., reification/qualifiers/provenance nodes), which is excluded by definition of the baseline. Hence the baseline is not contextually distinguishable.
\end{proof}

\begin{theorem}[AEG is strictly more comprehensive than predicate-labeled KG]
\label{thm:aeg_strict_comprehensive}
There exists an AEG that is contextually distinguishable, while no predicate-labeled KG baseline can represent the same information
without extending the formalism beyond predicate-typed edges.
\end{theorem}

\begin{proof}
Consider two atoms $a_1,a_2\in\mathcal{A}$ that share the same relational core $(h,r,t)$ but differ in contextual semantics:
$a_1$ asserts that $r(h,t)$ holds under context $c_1$, and $a_2$ asserts that $r(h,t)$ holds under a different context $c_2$ with $c_1\neq c_2$
(where contexts may encode time windows, scope, attribution, or any discourse-resolved qualifiers).
By the atom definition in AtomicRAG, each $a_i$ is a minimal self-contained proposition and can be stored as a distinct semantic object in AEG.
Both atoms connect to entities via containment edges, so both remain retrievable and composable.

In contrast, the predicate-labeled KG baseline must collapse both evidences to the same typed edge $(h,r,t)$, losing the distinction between $c_1$ and $c_2$,
by Lemma~\ref{lem:kg_not_contextual}. Therefore the information represented by this AEG cannot be represented in the baseline KG without extending the formalism.
This proves strict comprehensiveness.
\end{proof}

\subsubsection{Robustness: decoupled semantics reduces propagation leakage induced by noisy predicate edges}

\paragraph{Propagation model.}
Let $P$ denote the row-normalized transition matrix used in personalized PageRank (PPR).
Given a personalization vector $\boldsymbol{\pi}$, PPR satisfies
\[
\mathbf{r}=\rho\,\boldsymbol{\pi}+(1-\rho)P^\top\mathbf{r},\qquad \rho\in(0,1).
\]
Partition nodes into a relevant region $R$ and an irrelevant region $I$ (depending on the query).
Consider a two-region macro transition matrix
\[
T=
\begin{pmatrix}
1-\gamma & \gamma\\
\varepsilon & 1-\varepsilon
\end{pmatrix},
\qquad e=(1,0),
\]
where $\gamma$ is the probability of leaving $R$ to $I$ via cross-region edges under the random walk,
and $\varepsilon$ is the probability of returning from $I$ to $R$.

\begin{lemma}[Relevant mass under two-region PPR and monotone leakage]
\label{lem:ppr_leakage_appendix}
Let $\varphi=(\varphi_R,\varphi_I)$ be the stationary distribution over $\{R,I\}$ induced by PPR on the macro chain.
Then
\[
\varphi_R=\frac{\rho+(1-\rho)\varepsilon}{\rho+(1-\rho)(\gamma+\varepsilon)}.
\]
In particular, if $\varepsilon\approx 0$, then
\[
\varphi_R \approx \frac{\rho}{\rho+(1-\rho)\gamma},
\]
and $\varphi_R$ is strictly decreasing in $\gamma$.
\end{lemma}

\begin{proof}
The macro PPR fixed-point equation is $\varphi=\rho e+(1-\rho)\varphi T$ with $\varphi_R+\varphi_I=1$.
Solving the first coordinate yields the closed form. Differentiating w.r.t.\ $\gamma$ gives a strictly negative derivative,
hence $\varphi_R$ decreases strictly with $\gamma$.
\end{proof}

Lemma~\ref{lem:ppr_leakage_appendix} shows that robustness of propagation-based retrieval reduces to controlling the cross-region leakage parameter $\gamma$.

\begin{assumption}[Predicate-edge noise induces larger leakage than containment backbone]
\label{assump:noise}
For a fixed corpus and extraction pipeline, spurious predicate-typed edges introduce cross-region transitions at least as often as containment edges do.
Moreover, in AEG, auxiliary entity--entity edges (including association or synonym links) are down-weighted so that their total transition probability
contribution is bounded by a factor $\beta\in(0,1)$ relative to the backbone containment transitions.
\end{assumption}

\begin{theorem}[AEG yields smaller propagation leakage than predicate-labeled KG]
\label{thm:aeg_more_robust}
Under Assumption~\ref{assump:noise}, the effective cross-region leakage parameter induced by AEG satisfies
$\gamma_{\mathrm{AEG}}\le \gamma_{\mathrm{KG}}$ for the predicate-labeled KG baseline.
Consequently, the relevant stationary mass satisfies $\varphi_R^{\mathrm{AEG}}\ge \varphi_R^{\mathrm{KG}}$, with strict inequality when the leakages differ.
\end{theorem}

\begin{proof}
In the predicate-labeled KG baseline, transitions between entity nodes are directly realized by predicate-typed edges.
Extraction errors may create spurious cross-region edges, which contribute fully to the random-walk probability of leaving $R$, increasing $\gamma_{\mathrm{KG}}$.

In AEG, the backbone transitions are mediated by containment edges:
a transition from an entity $e$ to an atom $a$ requires $e\in\mathcal{E}(a)$ and a transition from $a$ to another entity $e'$ requires $e'\in\mathcal{E}(a)$.
These transitions are structurally constrained by mention locality, and by Assumption~\ref{assump:noise} their propensity to induce cross-region jumps
is no larger than that of spurious predicate edges. Auxiliary entity--entity edges are down-weighted so their contribution is bounded by $\beta$.
Therefore the total probability mass assigned to leaving $R$ in AEG is no larger than that in the predicate-labeled KG baseline, implying
$\gamma_{\mathrm{AEG}}\le \gamma_{\mathrm{KG}}$.

Finally, Lemma~\ref{lem:ppr_leakage_appendix} shows $\varphi_R$ decreases monotonically with $\gamma$,
hence $\varphi_R^{\mathrm{AEG}}\ge \varphi_R^{\mathrm{KG}}$, with strict inequality when $\gamma_{\mathrm{AEG}}<\gamma_{\mathrm{KG}}$.
\end{proof}

Combining Theorem~\ref{thm:aeg_strict_comprehensive} (strict comprehensiveness) and Theorem~\ref{thm:aeg_more_robust} (propagation robustness)
establishes Proposition~1.

\subsection{Proof of Proposition 2: Granularity alignment facilitates retrieval}
\label{proof2}

This section proves a two-sided granularity mismatch principle and shows why AtomicRAG's atom-level storage and query decomposition
move retrieval into a favorable regime.

\subsubsection{Unified formalization: retrieval as ranking over evidence sets}

Let $\mathcal{A}$ denote the universe of atomic evidence items (atoms in AtomicRAG).
For a query $q$, assume there exists a minimal sufficient evidence set $A^{\ast}(q)\subseteq \mathcal{A}$ necessary to support the correct answer,
and denote $m:=|A^{\ast}(q)|$.
A retrieval unit $U$ (e.g., chunk, subgraph, community, path, triple, or atom) serializes into an evidence set $C(U)\subseteq\mathcal{A}$.
Define
\[
r(U):=|A^{\ast}(q)\cap C(U)|,\qquad M(U):=|C(U)|.
\]
Define coverage and purity:
\[
\mathrm{Cov}(q,U)=\frac{r(U)}{m},\qquad \mathrm{Pur}(q,U)=\frac{r(U)}{M(U)}.
\]

\subsubsection{Coarse units dilute: separation scales with purity and misranking worsens with set size}

Assume an atom-level scoring model $s(q,a)$:
\[
s(q,a)=\mu_{Y(a)}+\varepsilon_a,\qquad Y(a)=\mathbf{1}[a\in A^{\ast}(q)],
\]
where $\mu_1>\mu_0$, $\Delta\mu=\mu_1-\mu_0>0$, and $\varepsilon_a$ are i.i.d.\ $\sigma$-sub-Gaussian.
To rank units, use mean aggregation:
\[
S(q,U)=\frac{1}{M(U)}\sum_{a\in C(U)} s(q,a).
\]

\begin{lemma}[Expected score gap equals purity-scaled signal]
\label{lem:gap_expect_appendix}
Let $U^+$ have $r:=r(U^+)\ge 1$ and $M:=M(U^+)$, and let $U^-$ satisfy $r(U^-)=0$ and $M(U^-)=M$.
Then
\[
\mathbb{E}[S(q,U^+)]-\mathbb{E}[S(q,U^-)]=\frac{r}{M}\Delta\mu.
\]
\end{lemma}

\begin{proof}
Expanding $S(q,U)$ and taking expectations removes the zero-mean noise terms, leaving the difference in means on necessary atoms,
scaled by the fraction $r/M$ of necessary atoms contained in the unit.
\end{proof}

\begin{theorem}[Misranking probability bound degrades with coarse evidence sets]
\label{thm:misrank_bound_appendix}
Under the same setting as Lemma~\ref{lem:gap_expect_appendix},
\[
\mathbb{P}\big(S(q,U^-)\ge S(q,U^+)\big)\le \exp\!\left(-\frac{r^2(\Delta\mu)^2}{4\sigma^2\,M}\right).
\]
\end{theorem}

\begin{proof}
Let $D=S(q,U^+)-S(q,U^-)=\frac{r}{M}\Delta\mu+X$ where $X$ is the difference of two independent averages of sub-Gaussian noises.
By sub-Gaussian closure, $X$ is $(\sqrt{2}\sigma/\sqrt{M})$-sub-Gaussian.
A standard tail bound yields
\[
\mathbb{P}(D\le 0)=\mathbb{P}\!\left(X\le -\frac{r}{M}\Delta\mu\right)\le \exp\!\left(-\frac{r^2(\Delta\mu)^2}{4\sigma^2\,M}\right).
\]
\end{proof}

Theorem~\ref{thm:misrank_bound_appendix} shows that coarse retrieval units (large $M$ with small $r$) incur weak separation and unstable ranking.

\subsubsection{Fine units fail by coverage limits under top-$k$}

\begin{theorem}[Coverage upper bound for overly fine units under top-$k$]
\label{thm:coverage_upper_appendix}
Assume $|C(U)|\le c$ for every retrieval unit $U$ (e.g., $c=1$ for triple-level units).
For any top-$k$ set $\{U_1,\ldots,U_k\}$,
\[
\left|A^{\ast}(q)\cap \bigcup_{i=1}^k C(U_i)\right|\le kc,
\qquad
\mathrm{Cov}\big(q,\{U_i\}_{i=1}^k\big)\le \frac{kc}{m}.
\]
In particular, if $kc<m$, then full coverage is impossible regardless of ranking quality.
\end{theorem}

\begin{proof}
Since $A^{\ast}(q)\cap \cup_{i=1}^k C(U_i)\subseteq \cup_{i=1}^k C(U_i)$,
\[
\left|A^{\ast}(q)\cap \bigcup_{i=1}^k C(U_i)\right|
\le \left|\bigcup_{i=1}^k C(U_i)\right|
\le \sum_{i=1}^k |C(U_i)|\le kc.
\]
Dividing by $m$ yields the bound on coverage.
\end{proof}

\subsubsection{Why AtomicRAG improves the regime: atom-level units plus query decomposition}

AtomicRAG aligns granularity by (i) using self-contained atoms as semantic carriers, avoiding coarse-unit dilution, and
(ii) decomposing complex queries into a small set of atomic sub-queries, reducing the effective evidence demand per retrieval instance.

Let the effective query set be $\widetilde{\mathcal{Q}}(q)$ as defined in the main text.
For LaTeX robustness, we introduce a shorthand:
\[
\widetilde{\mathcal{Q}}_q := \widetilde{\mathcal{Q}(q)}.
\]
For each $q'\in \widetilde{\mathcal{Q}}_q$, let $A^{*}(q')$ denote the minimal sufficient evidence set with size $m_{q'}=|A^{*}(q')|$.
Define the overall target evidence as
\[
A^{*}_{\mathrm{all}}(q)=\bigcup_{q'\in\widetilde{\mathcal{Q}}_q} A^{*}(q').
\]

\begin{corollary}[Decomposition relaxes top-$k$ coverage constraints]
\label{cor:decomp_coverage}
Suppose retrieval is performed independently for each $q'\in\widetilde{\mathcal{Q}}_q$ with the same top-$k$ budget and unit-size bound $|C(U)|\le c$.
Then full coverage for each sub-query requires only $kc\ge m_{q'}$ rather than $kc\ge |A^{*}_{\mathrm{all}}(q)|$.
Thus, when $\max_{q'\in\widetilde{\mathcal{Q}}_q} m_{q'} \ll |A^{*}_{\mathrm{all}}(q)|$, decomposition enlarges the feasible region of full-coverage retrieval.
\end{corollary}

\begin{proof}
Apply Theorem~\ref{thm:coverage_upper_appendix} to each sub-query $q'$ separately.
The necessary condition for full coverage becomes $kc\ge m_{q'}$ for each $q'$.
If decomposition reduces the maximum evidence demand per sub-query, the constraint is relaxed accordingly.
\end{proof}

Finally, because atoms are defined as minimal self-contained propositions, AtomicRAG can keep $M(U)$ small for each semantic unit while maintaining non-trivial
overlap $r(U)$ for relevant units. This increases purity $r(U)/M(U)$ and strengthens the misranking exponent in
Theorem~\ref{thm:misrank_bound_appendix}. This establishes Proposition~2.


\section{Case Study}
\label{app:case}
\begin{figure}[t]
  \centering
  \includegraphics[page=1,width=\linewidth]{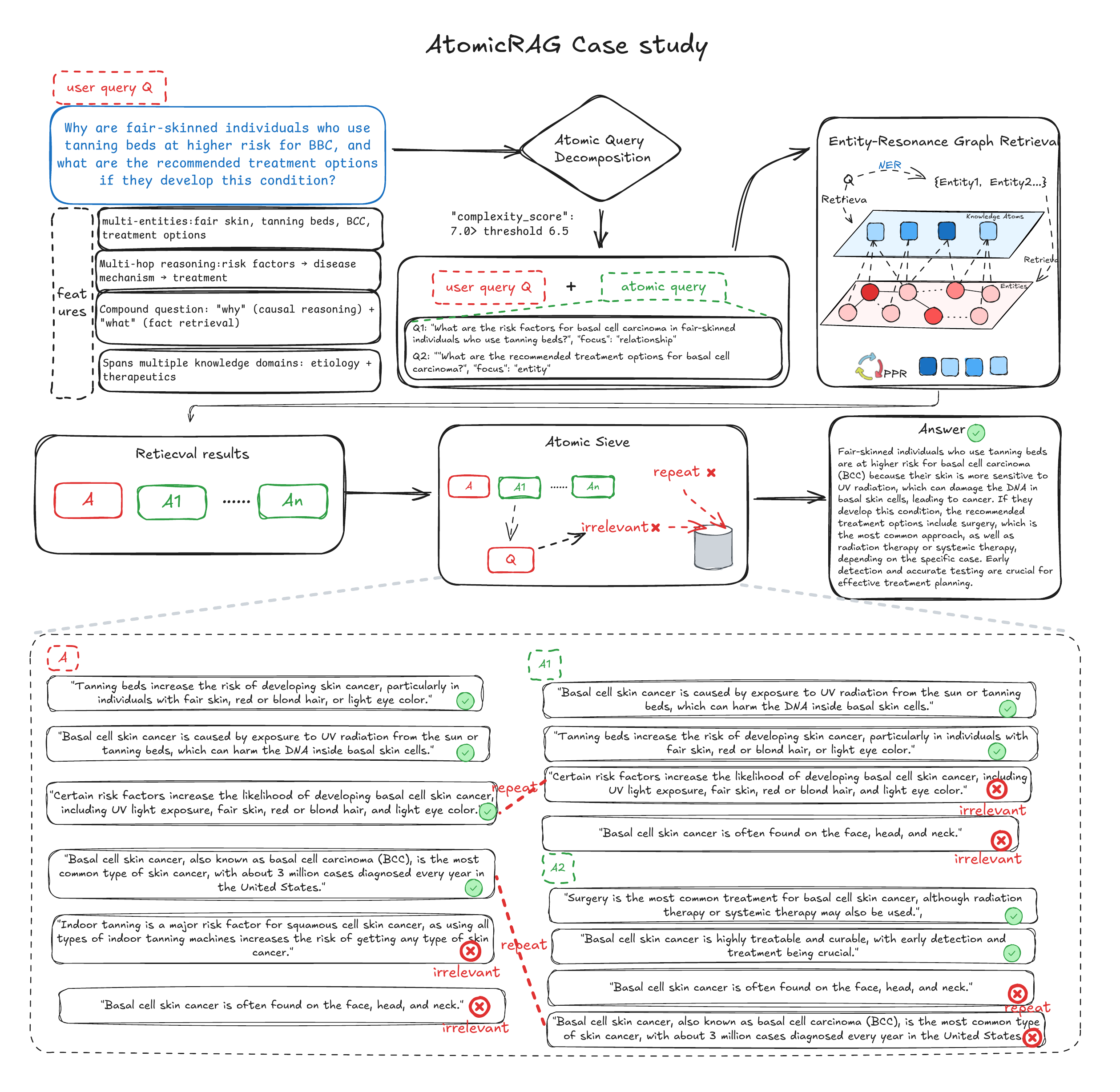}
  \caption{Case study.}
  \label{fig:case-study}
\end{figure}
Figure~\ref{fig:case-study} illustrates a compound user query that jointly asks for a causal explanation (``why'') and actionable recommendations (``what'') in a setting with multiple salient entities (fair skin, tanning beds, basal cell carcinoma (BCC), treatment options) and an implicit multi-hop chain (risk factors $\rightarrow$ mechanism $\rightarrow$ treatment). In a standard chunk-based RAG pipeline, a single retrieval pass tends to aggregate heterogeneous evidence from etiology and therapeutics into one context window, which amplifies redundancy and introduces topic drift (e.g., prevalence statistics or anatomical distribution), ultimately weakening causal grounding and obscuring the treatment-centric part of the question.

AtomicRAG addresses this failure mode by explicitly separating \emph{knowledge indexing} (which evidence belongs to which sub-intent) from \emph{knowledge representation} (how evidence units are encoded and deduplicated). The system first estimates query complexity; in this example, the score exceeds the decomposition threshold (7.0 $>$ 6.5), triggering \emph{Atomic Question Decomposition}. The query is split into two atomic sub-queries with distinct foci: $Q_1$ targets the \emph{relationship} between tanning bed use, fair skin, and BCC risk; $Q_2$ targets \emph{entity-centric} treatment options for BCC. This decomposition prevents cross-domain interference by construction: etiological evidence is retrieved and synthesized for $Q_1$, while therapeutic evidence is retrieved and synthesized for $Q_2$.

For each atomic query, \emph{Entity-Resonance Graph Retrieval} performs NER to obtain an entity set and then runs a graph-based propagation (e.g., Personalized PageRank) over an entity--atom graph to surface candidate \emph{knowledge atoms}. Compared to chunk retrieval, atom-level evidence units improve controllability during downstream aggregation because they are both smaller (reducing irreducible noise) and explicitly indexed by entities (supporting compositional evidence tracing). However, graph retrieval can still surface near-duplicates and weakly related atoms. AtomicRAG therefore applies an \emph{Atomic Sieve} that (i) deduplicates repeated atoms (e.g., multiple paraphrases stating that indoor tanning elevates skin cancer risk in fair-skinned individuals), and (ii) filters atoms that are off-target for the current sub-query (e.g., incidence counts, common anatomical sites, or statements about other skin cancer types). After sieving, the retained atoms for $Q_1$ concentrate on the causal pathway linking UV exposure from tanning beds to DNA damage in basal cells and heightened susceptibility in fair skin, while the retained atoms for $Q_2$ focus on treatment decisions (surgery as the common first-line option, with radiation or systemic therapy depending on case factors) and the role of early detection/testing for planning.

Finally, the generator composes the response by \emph{fusing} the two curated atom sets, yielding an answer that cleanly separates causal justification from treatment recommendations while maintaining an explicit evidence chain. This example highlights the core benefit of AtomicRAG in practice: by representing evidence as atomic facts and organizing retrieval around decomposed intents, the system reduces redundancy, limits domain leakage, and stabilizes multi-hop reasoning under compound query objectives.

\section{Prompt Templates}
\label{app:prompts}

AtomicRAG relies on a small set of prompt templates that operationalize (i) corpus-to-structure construction, (ii) query decomposition and evidence selection, and (iii) answer generation and evaluation. To keep the appendix readable, we only summarize the role of each prompt family here, and provide the full templates verbatim as figures for exact reproducibility (see Figures~\ref{fig:prompt-ner}--\ref{fig:prompt-qa2}).

\paragraph{Named Entity Recognition.}
We use an entity extraction prompt to identify salient named entities in each passage. The output is a structured JSON list, which is then reused by downstream extraction prompts to encourage entity-grounded triples and fragments ( Figure~\ref{fig:prompt-ner}).

\paragraph{Unified Triple \& Knowledge Atom Extraction.}
This prompt jointly extracts (a) RDF-style triples and (b) self-contained knowledge atoms from the same passage, with explicit constraints to resolve coreference, preserve quantities/time spans, and avoid redundant fragments. It returns a single JSON object containing triples, atoms, and atom-level entity mentions ( Figure~\ref{fig:prompt-unti}).

\paragraph{Atomic Question Decomposition.}
We adopt a single-call decomposition prompt that first scores question complexity and then (only when needed) produces a small set of atomic sub-questions with focus tags. This ensures decomposition is used conservatively and remains retrieval-actionable (Figure~\ref{fig:prompt-query}).

\paragraph{Knowledge Atom Filtering.}
Given the user question and a candidate set of knowledge atoms, the filter prompt selects only the atom IDs that are directly relevant for answering the question. The output is restricted to an index list in JSON, preventing the model from inventing new evidence (Figure~\ref{fig:prompt-filter}).

\paragraph{Abstract QA and precise QA.}
We use two reading-comprehension prompts for answer synthesis: (i) an \emph{abstract} QA prompt that produces a complete, self-contained answer with brief evidence-based reasoning, and (ii) a \emph{precise} QA prompt that outputs a concise final answer when the benchmark expects short-form responses (Figures~\ref{fig:prompt-qa} and \ref{fig:prompt-qa2}).
\section{Limitations}
A limitation of AtomicRAG is that its effectiveness depends on the quality and consistency of the offline atom–entity graph construction and the online query decomposition. In particular, the atomization step and entity canonicalization are typically produced by an instruction-tuned LLM (or an automated pipeline), so variations in extraction quality seeps into the downstream graph connectivity and retrieval neighborhoods. Likewise, when atomic question decomposition is enabled, the decomposition granularity and sub-question coverage can influence how well the subsequent propagation and sieve steps identify the right evidence chain. While our design reduces reliance on predicate-labeled edges and is generally robust in noisy settings, it may be less expressive for tasks that require strict relation semantics (e.g., fine-grained temporal/causal constraints), where additional relation-aware signals could further help. Finally, AtomicRAG introduces extra system components (graph storage, indexing, and optional decomposition), and although we report efficiency results, scaling to continuously updated corpora (frequent insertions/deletions) may require additional engineering for incremental updates and stability.

\begin{figure}[t]
  \centering
  \includegraphics[page=1,width=\linewidth]{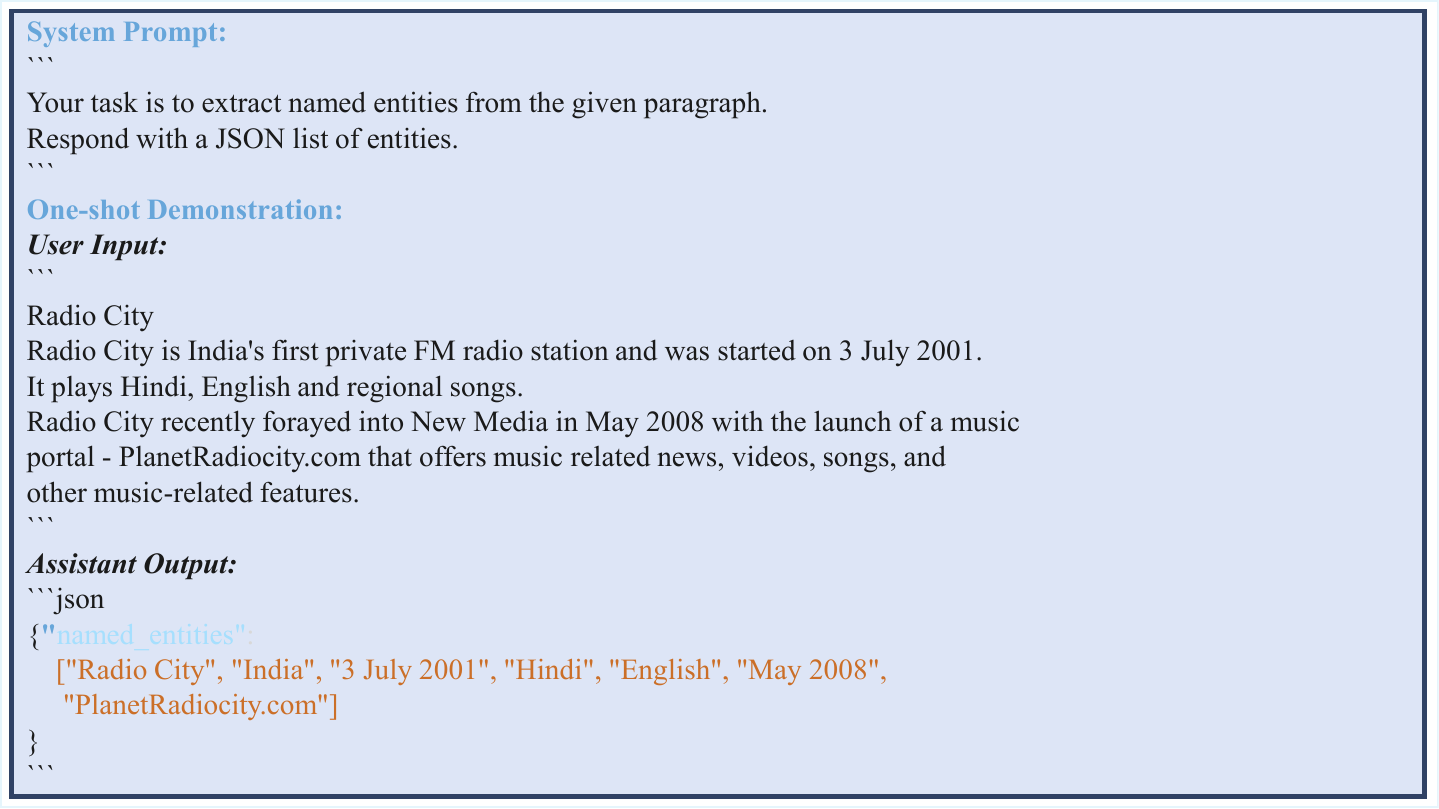}
  \caption{Prompt template for named entity recognition (NER).}
  \label{fig:prompt-ner}
\end{figure}

\begin{figure}[t]
  \centering
  \includegraphics[page=1,width=\linewidth]{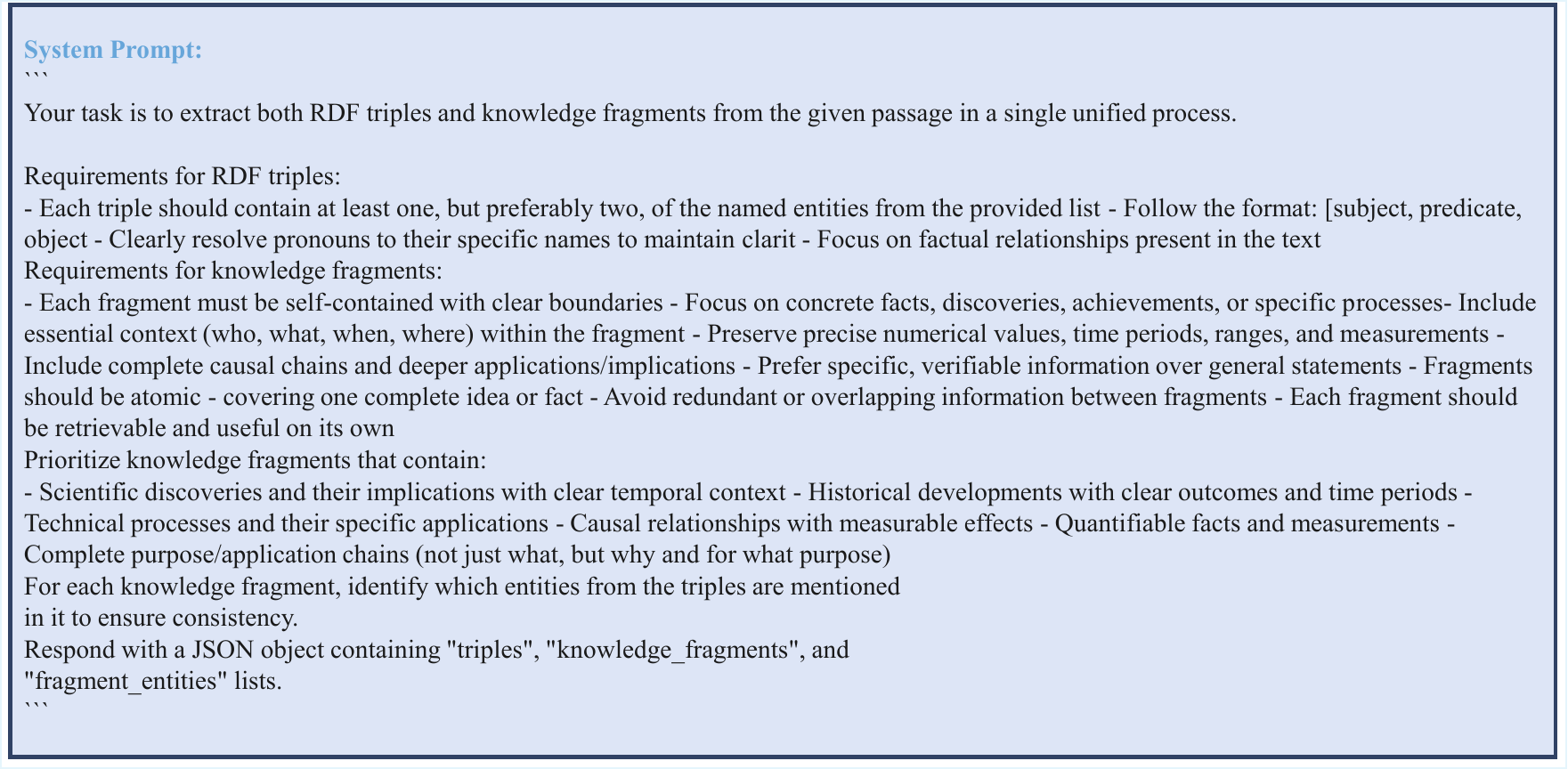}
  \caption{Prompt template for unified triple and knowledge atom extraction.}
  \label{fig:prompt-unti}
\end{figure}

\begin{figure}[t]
  \centering
  \includegraphics[page=1,width=\linewidth]{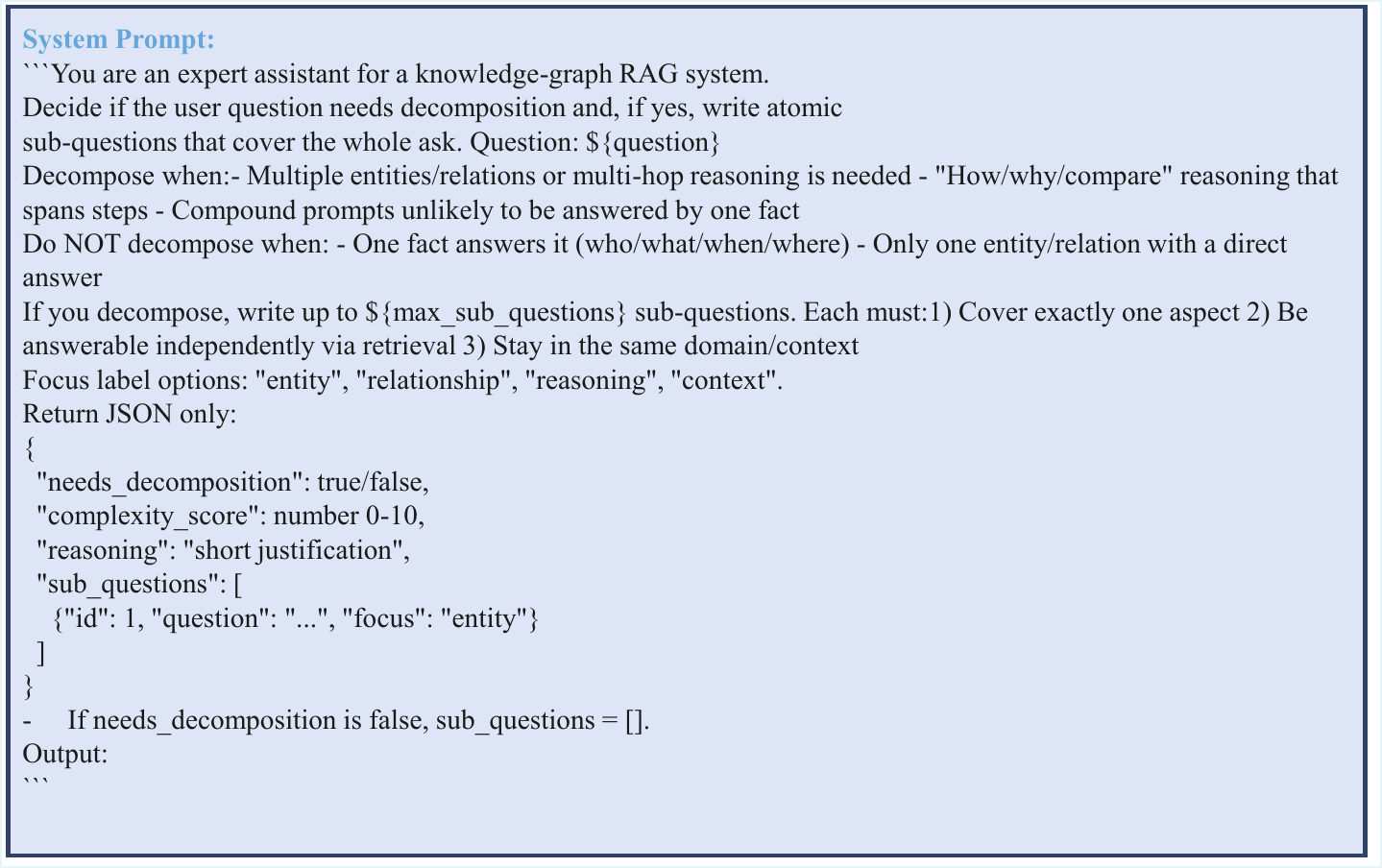}
  \caption{Prompt template for question complexity scoring and atomic decomposition.}
  \label{fig:prompt-query}
\end{figure}

\begin{figure}[t]
  \centering
  \includegraphics[page=1,width=\linewidth]{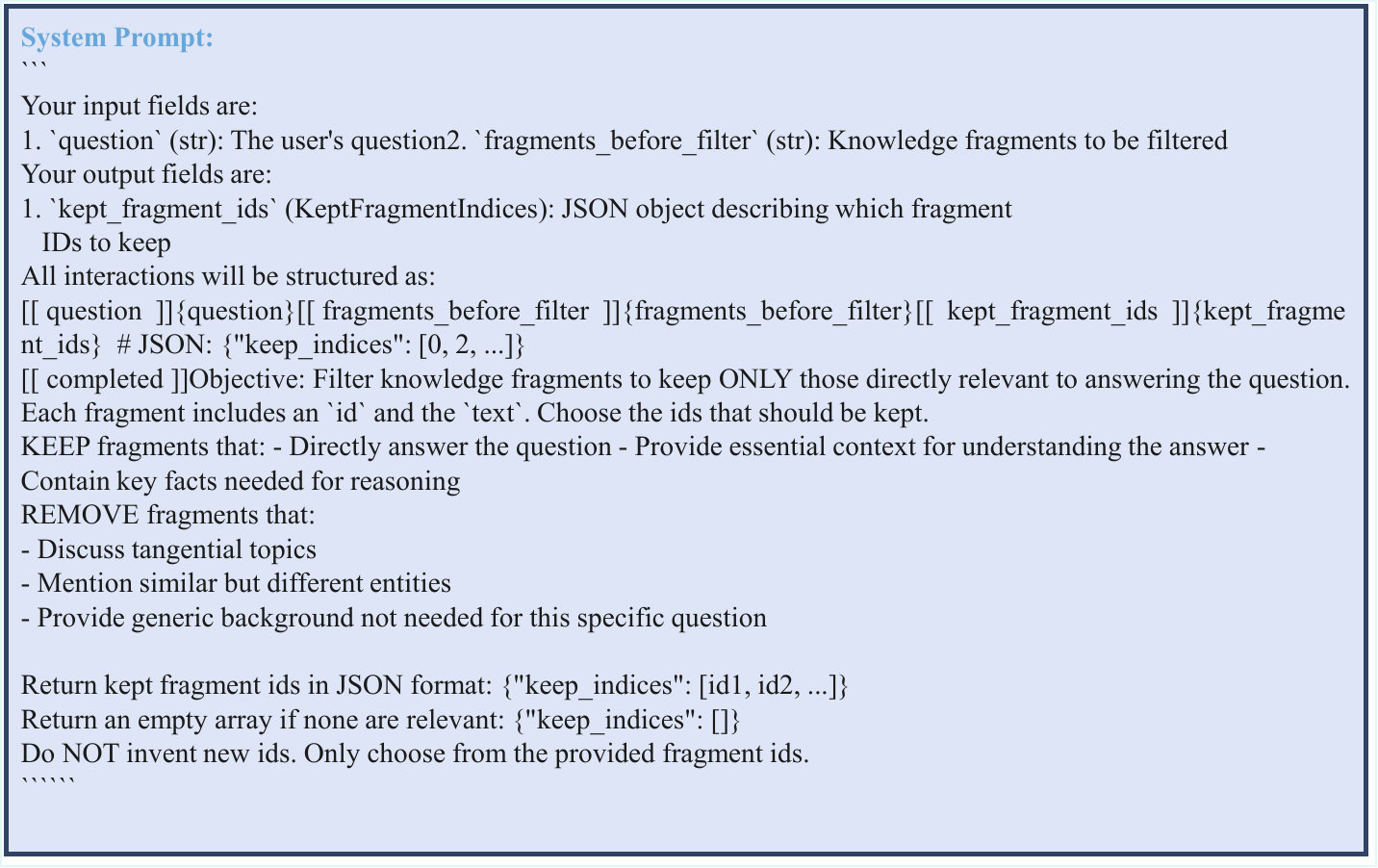}
  \caption{Prompt template for knowledge atom filtering.}
  \label{fig:prompt-filter}
\end{figure}

\begin{figure}[t]
  \centering
  \includegraphics[page=1,width=\linewidth]{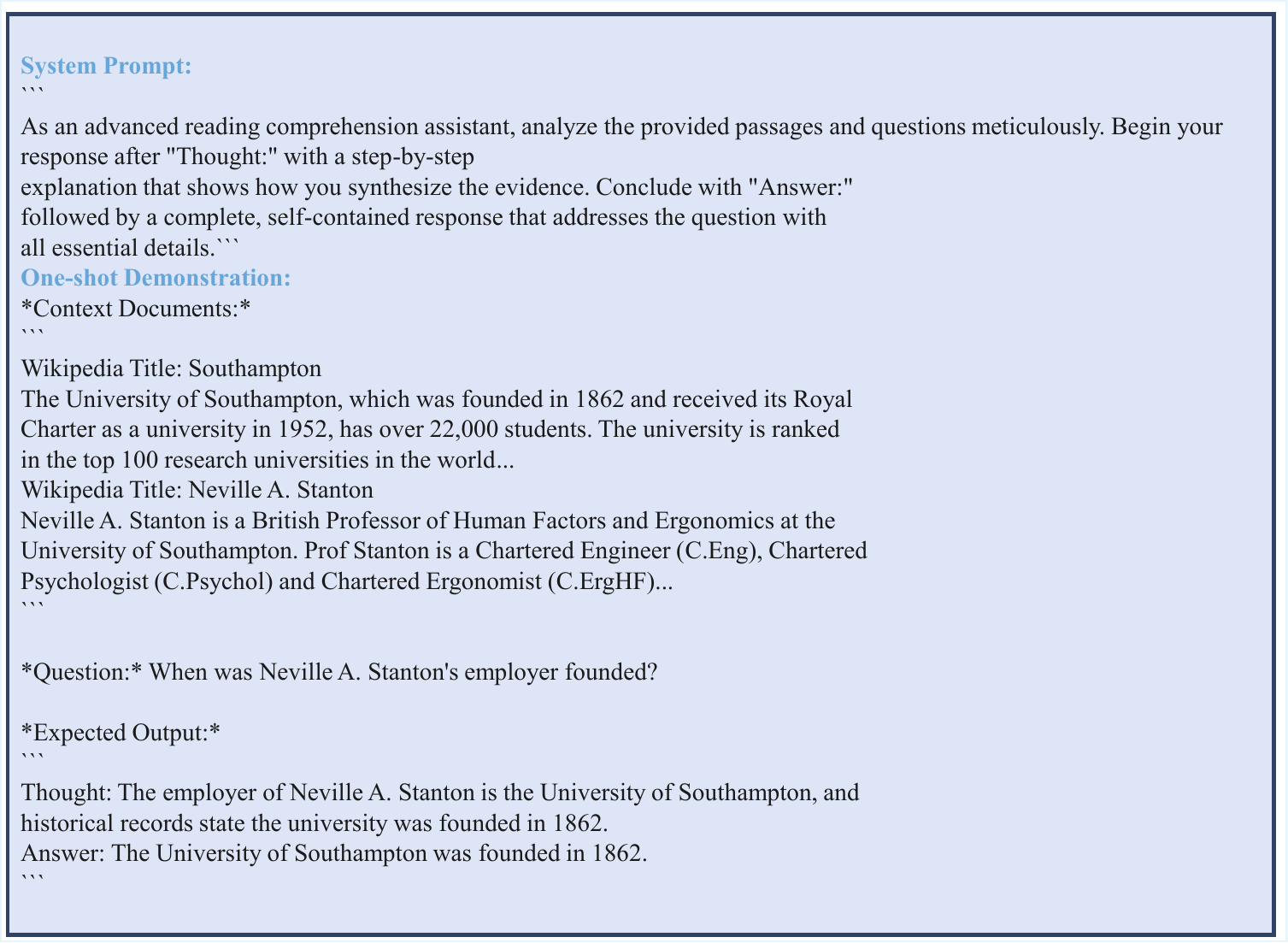}
  \caption{Prompt template for abstract question answering.}
  \label{fig:prompt-qa}
\end{figure}

\begin{figure}[t]
  \centering
  \includegraphics[page=1,width=\linewidth]{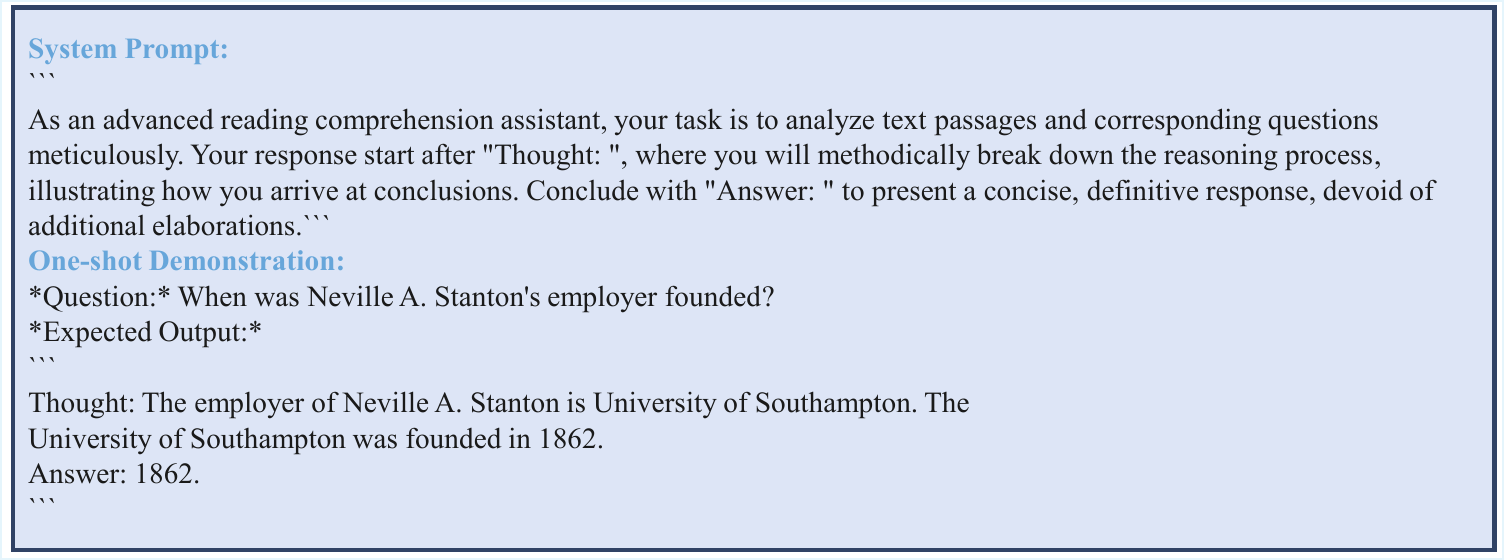}
  \caption{Prompt template for precise question answering.}
  \label{fig:prompt-qa2}
\end{figure}


\end{document}